%% file: neurips_2026.tex
\documentclass{article}

\PassOptionsToPackage{numbers}{natbib}
\usepackage[main, final]{neurips_2026}

\usepackage[utf8]{inputenc} 
\usepackage[T1]{fontenc}    
\usepackage{hyperref}       
\usepackage{url}            
\usepackage{booktabs}       
\usepackage{amsfonts}       
\usepackage{nicefrac}       
\usepackage{microtype}      
\usepackage{xcolor}         
\usepackage{graphicx} 
\usepackage{amsmath}
\usepackage{amssymb}
\usepackage{amsthm}
\usepackage{multirow}
\usepackage{float}
\usepackage{subfig}
\usepackage{listings} 
\theoremstyle{plain}
\newtheorem{theorem}{Theorem}[section]

\theoremstyle{definition}
\newtheorem{definition}[theorem]{Definition}
\newtheorem{assumption}[theorem]{Assumption}
\theoremstyle{remark}
\newtheorem{remark}[theorem]{Remark}

\DeclareMathOperator*{\argmax}{argmax}

\usepackage[ruled]{algorithm2e}
\usepackage{wrapfig}
\lstdefinestyle{promptstyle}{
    language=Tex,          
    basicstyle=\small\ttfamily,  
    backgroundcolor=\color{gray!5}, 
    frame=single,          
    frameround=tttt,       
    numbers=left,          
    numberstyle=\tiny,     
    breaklines=true,       
    columns=fullflexible,  
    xleftmargin=2em,       
    xrightmargin=2em       
}
\title{Towards High Semantic Fidelity: Hyperdimensional Symbolic Messages in Multi-Agent Communication}

\author{%
\textbf{Shoucheng Song}\textsuperscript{1,2,3}, 
\textbf{Chang Yao}\textsuperscript{1,2},
  \textbf{Hao Wu}\textsuperscript{1,2}, 
  \textbf{Youfang Lin}\textsuperscript{1,2}, 
  \textbf{Sheng Han}\textsuperscript{1,2}, 
  \textbf{Kai Lv}\textsuperscript{1,2}\thanks{Corresponding author.} \\
  \textsuperscript{1}School of Computer Science \& Technology,
  Beijing Jiaotong University, Beijing, China \\
  \textsuperscript{2}Beijing Key Laboratory of Traffic Data Mining and
  Embodied Intelligence, Beijing, China \\
  \textsuperscript{3}Institute of Artificial Intelligence of China Telecom (TeleAI), Beijing, China\\
  \texttt{\{insis\_songsc,yaochang,insis\_wh,yflin,shhan,lvkai\}@bjtu.edu.cn}
}

\begin{document}

\maketitle

\begin{abstract}
Semantic fidelity is essential for multi-agent communication under non-ideal channels, comprising transmission fidelity and conversion fidelity.
Yet existing numerical and neural message formats cannot simultaneously achieve both. 
To address this, we propose a novel message format-\textbf{H}yperdimensional \textbf{S}ymbolic \textbf{M}essages (HSM), and its corresponding communication framework \textbf{Herm}. 
For transmission fidelity, Herm suppresses semantic entanglement via an entity-attribute-value structured encoding pipeline.
For conversion fidelity, Herm first leverages a receiver-centric ego-binding mechanism to align semantics. Then, a similarity-aware dual-branch fusion module is proposed to prevent critical semantic dilution.
To realize explicit message-semantic conversion, 
we devise an iterative decoding scheme via algebraic inverse operation.
Experiments exhibit that Herm outperforms baselines across various noisy scenarios while maintaining superior structural readability.
\end{abstract}

\section{Introduction}

In multi-agent reinforcement learning (MARL) \cite{yao26direct,AIM,chu2026value,yao1,jiang2026budget,yao2}, communication serves as the foundation for emerging coordination.
In non-ideal channels, the effectiveness of coordination crucially depends on whether the semantics transmitted between agents can be preserved completely and understood accurately.
We refer to this critical capability as \textbf{Semantic Fidelity}, which features
(1) \textbf{Transmission Fidelity}:
The semantics of messages remain stable and undistorted in noisy or adversarial channels.
Even with the perturbed message, the receiver’s understanding remains consistent with the sender's original semantics.
(2) \textbf{Conversion Fidelity}:
There exists a bidirectional conversion between agents’ native semantics (intentions, trajectories, etc.) and messages.
Receivers can recover the native semantics from messages without black-box ambiguity, and transform messages into decision-valid semantics.

As shown in Fig. \ref{TSNE}, existing communication methods fall into two paradigms, yet \textit{neither can fully satisfy both fidelity}.
\textbf{Numerical messages} \cite{vbc,TMC,MAIC}, typically Q-value lists, support basic bidirectional conversion due to their simple structure and limited semantic space.
But even mild channel interference can cause catastrophic semantic distortion.
\textbf{Neural messages} \cite{tarmac,NDQ,CACOM,T2MAC,CODE,TGCNET,dear} enable improved robustness via tailored losses \cite{CroMAC2024,MDGI2024}.
However, the black-box encoding entirely fails to support conversion fidelity.
Neither agents nor humans can recover the native semantics from neural messages.
This inherent dilemma motivates the core question investigated in this paper:
\textbf{\textit{Can we design a novel message format in MARL that achieves high semantic fidelity?}}

Based on the above requirement, Hyperdimensional Computing (HDC) \cite{HRR,BSC} theoretically offers an ideal mathematical framework. The quasi-orthogonality of high-dimensional random vectors inherently aligns with the noise-robust nature of  ``transmission fidelity''. Meanwhile, its invertible algebraic operations enable ``conversion fidelity'' by allowing the extraction of original semantics. 

However, \textbf{in practice}, applying HDC directly to multi-agent communication faces profound challenges that compromise the semantic fidelity.
 \textbf{For transmission fidelity}, standard HDC's flat encoding requires extremely high dimensions when handling MARL's multi-entity perception. With practical low-dimensional settings, HDC compromises the quasi-orthogonality and causes semantic entanglement.
\textbf{For conversion fidelity}, HDC can only transmit static ego information and fails to provide adaptation for different receivers. Thus, sender-centric encoding leads to semantic misalignment between broadcast messages and receivers' local decision contexts. Moreover, HDC's native equal-weight bundle inevitably dilutes critical decision-valid semantics.

To address the above issues that undermine semantic fidelity, we propose a novel structured algebraic message format: Hyperdimensional Symbolic Messages (HSM) and its corresponding communication framework, named Herm.
\textbf{To improve transmission fidelity}, Herm devises an entity-attribute-value structured encoding pipeline for the sender. Leveraging HDC's algebraic compositionality, this pipeline strictly aligns the perceptual logic of MARL with HDC operations. This alignment significantly enhances HSMs' semantic expressiveness and suppresses the semantic entanglement.
\textbf{To ensure conversion fidelity}, Herm enhances the decision-valid semantic conversion for the receiver from two perspectives: semantic misalignment and semantic dilution. First, Herm introduces a receiver-centric ego-binding mechanism. Without any learnable parameters, this mechanism implicitly projects globally broadcasted HSMs into the unique local semantic subspace of receivers, thereby achieving second-order semantics contextualization. Second, Herm designs a similarity-aware dual-branch fusion network to dynamically evaluate the relevance of messages in both perception supplementation and behavior understanding.

Furthermore, we design an iterative decoding scheme for HSMs to realize explicit message-to-semantics conversion. Building upon this, we explore paradigms for integrating HSMs with Large Language Models (LLMs).
Extensive experiments across multiple benchmarks (Hallway \cite{NDQ}, LBF \cite{LBF}, SMAC \cite{SMAC}, SMACv2 \cite{SMACV2}) demonstrate Herm's robustness against noise (transmission fidelity) and the readability advantage (conversion fidelity).
Our contributions are threefold:

\begin{figure}[t]
    \centering
    \includegraphics[width=1.0\textwidth]{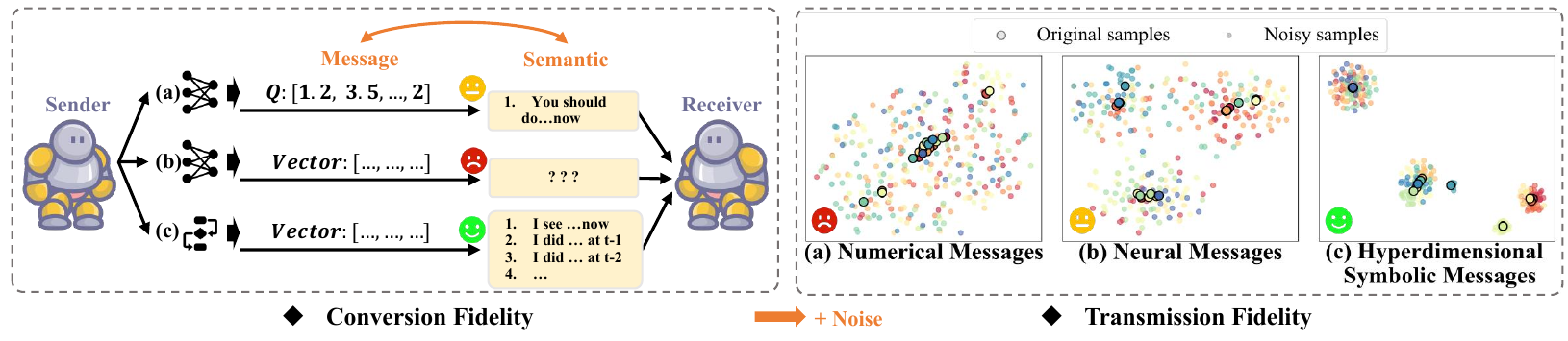}
    \caption{Comparison of Semantic Fidelity of Messages. The left shows the conversion fidelity in terms of the embedding process, message format, and semantic content. The right employs t-SNE to exhibit their transmission fidelity when exposed to noise.}
    \label{TSNE}
\end{figure}



\begin{itemize}
    \item We pioneer the definition of semantic fidelity in multi-agent communication and  the formalization of hyperdimensional symbolic messages.
    \item We design an entity-attribute-value structured encoding pipeline to eliminate semantic entanglement and guarantee transmission fidelity.
    \item We devise an ego-binding mechanism and similarity-aware dual-branch fusion to ensure decision-valid conversion fidelity.
\end{itemize}

\section{Related Works}
\label{gen_inst}


\textbf{Multi-Agent Robust Communication}. One scheme aims to identify malicious messages and mitigate their adverse effects. 
Early work employed Gaussian processes to model the message distribution, utilizing posterior probabilities for confidence evaluation \cite{GP2020}. 
$\mathfrak{R}$-MACRL \cite{MIS2022} implemented a two-stage architecture that first detects malicious messages, then performs restoration. 
AME \cite{AME2023} utilized diverse combinations of ablated messages to mitigate the impact of attacks. 
ADMAC \cite{ADMAC2024} leveraged dynamic labeling and binary classifier training to identify malicious messages. 
Some works trained independent channel adversaries to enhance robustness through adversarial training \cite{MA3C2024,DMAC2025}.
They typically involve additional communication stages that increase the training and deployment costs.
Furthermore, some studies focus on the robustness of messages via graph information bottleneck \cite{MDGI2024} or multi-view certification \cite{CroMAC2024}.
In contrast, we make a fundamental departure by replacing neural messages with hyperdimensional symbolic messages, achieving robustness (transmission fidelity) and readability (conversion fidelity).

\textbf{Hyperdimensional Computing in RL}. Inspired by HDC's applications in federated learning \cite{zhao2022fedhd} and classification tasks \cite{imani2018hierarchical}, some works have explored its adaptation to the RL domain. 
HDPG \cite{HDPG2022} and QHD \cite{QHD2023} established the first policy-based and value-based hyperdimensional reinforcement learning paradigms (HRL). 
Subsequent works have extended HRL to robotic control \cite{ROBOT2024}, cyber security \cite{Cyber2025}, and contextual bandit \cite{HDCB2025}.
However, the application of HDC in multi-agent systems remains unexplored. To the best of our knowledge, this work presents the first integration of HDC with multi-agent communication.

\section{Preliminaries}
\label{headings}

\textbf{Dec-POMDP.} Multi-agent collaboration is typically modeled as a decentralized partially observable Markov decision process (Dec-POMDP), defined by $<\mathcal{N},\mathcal{S},\mathcal{A},\mathcal{O},\mathcal{R},\mathcal{P},\gamma>$. $\mathcal{N}$ denotes multiple agents within the system. $\mathcal{S}$ and $\mathcal{A}$ represent the joint state and action space. $\mathcal{O}$ defines the mapping from joint states to local observations. $\mathcal{R}$ and $\mathcal{P}$ describe the reward and state transition function. $\gamma$ is the discount factor to scale the long-term return. In communication-enabled scenarios, Dec-POMDP can be extended with message space $\mathcal{M}$ to capture inter-agent communication. The objective in Dec-POMDP is to derive an optimal joint policy $\pi^*$ that maximizes the expected discounted return, formulated as $\pi^* = \operatorname*{argmax}_{\pi} \mathbb{E}_{\pi,\mathcal{P}}
    \left[
        \sum_{t=0}^{\infty} 
        \gamma^t \mathcal{R}(s_t, \boldsymbol{a}_t) 
    \right].$

\textbf{Hyperdimensional Computing.} Hyperdimensional computing (HDC) relies on randomly hypervectors (HVs) from a high-dimensional space. 
Due to the quasi-orthogonality of HVs, HDC provides a robust approach to information encoding, which inspires our work.
Each random HV can represent any single concept, such as an attribute, a numerical value, or a relational structure.
Through two fundamental operations—\textbf{Bind} ($\otimes$) and \textbf{Bundle} ($\oplus$)—HVs can symbolize composite concepts with increasing complexity. 
Bind computes a composite representation that is dissimilar to each input. 
Bundle adds multiple items together and retains similarity with all inputs. 
A simple illustration is shown as: (1) $  X=\text{``Color''} \otimes \text{``Red''}  \rightarrow \text{``Color = Red''}$, 
(2) $Y=\text{``Shape''} \otimes \text{``Circle''}  \rightarrow \text{``Shape = Circle''}$, 
(3) $\text{X} \oplus \text{Y}  \rightarrow \text{``Red Circle''}$.
HVs can then be mapped to specific semantics through unbinding and similarity matching. Based on this, we devise a message-to-semantic transformation to achieve conversion fidelity.
Here, we employ the Holographic Reduced Representations (HRR) \cite{HRR} as the HDC framework. Likewise, our method is compatible with other frameworks (Appendix \ref{append-hdc}).

\section{Methods}
\label{others}


The core of our work lies in two fundamental properties for semantic fidelity in communication: transmission fidelity and conversion fidelity, which cannot be traded off in existing numerical and neural message paradigms. 
In this paper, we design a novel message format, the hyperdimensional symbolic message (HSM), based on training-free hyperdimensional computing (HDC).

While HDC theoretically aligns with the two properties, standard HDC suffers from MARL-specific flaws that directly break the semantic fidelity.
Hence, we further propose our method \textbf{Herm}, as shown in Fig.\ref{framework}, which makes the corresponding optimization from the view of the sender and receiver.

\textbf{On the sender side}, Herm focuses on eliminating the inherent semantic entanglement flaw of standard flat HDC encoding.  We design a structured message encoding pipeline (Fig. \ref{framework} (b)) to balance semantic expressiveness and anti-interference capability with acceptable overhead.
\textbf{On the receiver side}, Herm targets the two core barriers in standard HDC — semantic misalignment and semantic dilution. We devise an ego-binding mechanism and similarity-aware dual-branch fusion (Fig. \ref{framework} (c)) to ensure decision-valid conversion fidelity during multi-source fusion. Beyond the training, we propose an additional iterative message decoding process for explicit message-to-semantic conversion.


\begin{figure*}[t]
\centering
\includegraphics[width=1.0\linewidth]{ 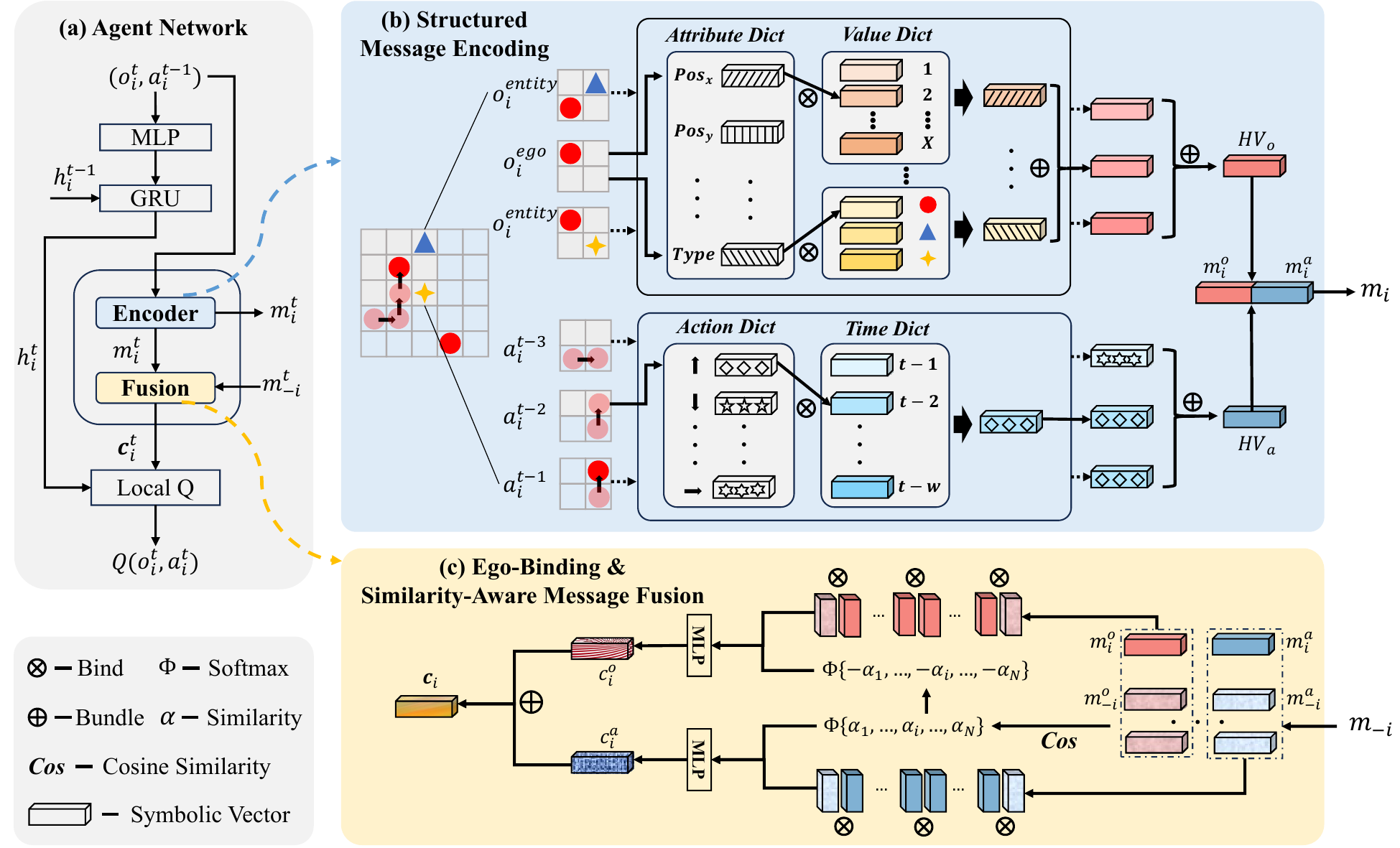}
\caption{Framework of Herm. (a) The agent network comprises a trajectory-encoding module and a communication module, with their outputs serving as inputs to the local Q-function. (b) On the sender side, the agent encodes its entity-wise observations and historical actions into HSM. (c) On the receiving side, the agent performs ego-binding and fuses messages based on similarity metrics.}
\label{framework}
\end{figure*}

\subsection{Structured Message Encoding for Transmission Fidelity}
\label{section-encode}
In the multi-agent system with limited communication bandwidth, standard flat HDC encoding could incur semantic entanglement.
Complex multi-entity perception violates HDC's quasi-orthogonality requirement, which is fundamental to transmission fidelity.
Hence, it is essential to explore the design of concise and efficient message structures for downstream decision-making.
 
We posit that communication enhances collaboration through two ways: \textbf {(1) supplementing subjective perception}, and \textbf{(2) interpreting teammates' objective behaviors}.
Given the stability of observations versus the stochastic nature of actions, our raw message $\tilde{m}_i^t$ contains instantaneous environmental perceptions $o_i^t$ alongside temporally-smoothed actions $\{a_i^{t-w},a_i^{t-w+1},\cdots,a_i^{t-1}\}$ where $w$ denotes the sliding window.
The original $\tilde{m}_i^t$ is then encoded into the symbolic message $m_i^t$ for broadcast through our structured encoding pipeline (Fig.\ref{framework} (b)).
Due to the refinement of $\tilde{m}_i^t$, we employ a modest hyperdimensional space (typically below 256), which incurs acceptable overhead.
\subsubsection{Perception Encoding via Observation Transformation.}
Agents typically represent their perception of the local environment through entity-wise attribute-value pairs: $o_i=\{\langle id, attr, val\rangle| id \in \mathcal{E}, attr \in \mathcal{A}, val \in \mathcal{V}\}$. $\mathcal{E}$ represents all entities (e.g., teammates, enemies). $\mathcal{A}$ denotes the attribute dictionary (e.g., position, health status) and $\mathcal{V}$ represents the corresponding value dictionary (e.g., X-coordinate, HP).
The attribute dictionary is first randomly instantiated in the hyperdimensional space and then processed via Schmidt orthogonalization to maintain strong orthogonality between attribute-aware $HV_{attr}$.
For the value dictionary involving continuous quantities, we employ a discretization preprocessing step followed by unified encoding of the discretized intervals.
Unlike mutually independent attributes, numerical values exhibit inherent correlations. We utilize cyclic replacement encoding to preserve these relationships. 
Let $val_{min}$ and $val_{max}$ denote the minimum and maximum values of an attribute. We first sample their hyperdimensional representations:
$HV_{min} \sim \mathbb{R}^d, HV_{max}\sim \mathbb{R}^d$.
For intermediate value $v \in (val_{min}, val_{max})$, its encoding is formulated as:
\begin{equation}
    HV_v = \boldsymbol{\Theta}(HV_{max}, HV_{min}, \left \lceil \rho \cdot d \right \rceil),
\label{number}
\end{equation}
where $\rho = (v - val_{min})/(val_{max} - val_{min})$ and $\boldsymbol{\Theta}(\cdot)$ replaces the first $\left \lceil \rho \cdot d\right \rceil$ dimensions of $HV_{min}$ with $HV_{max}$'s corresponding dimensions.
Our framework then converts agent observations into HSM via four key operations: (i) attribute-value binding ($\otimes$), (ii) attribute bundling ($\oplus$), (iii) identity binding ($\otimes$), and (iv) entity-wise bundling ($\oplus$), defined as:  
\begin{equation}
    HV_O = \sum_{j}^{\mathcal{E}}\oplus \left[HV_{id}^j \otimes \sum_{attr}^\mathcal{A} {\oplus(HV_{attr} \otimes HV_v)} \right],    
\label{obs_hd}
\end{equation}
where $\sum\oplus$ is the multi-bundle operation, and $HV_{id}^j$ denotes the identity vector for entity $j$. Although each agent bundles information about all entities, their encoding still differs due to partial observability.
Notably, (1) entity-wise observation partitioning is a common practice in MARL research \cite{ODIS, MATTAR, OPT, RPG} and introduces no manual prior knowledge; (2) while attributes carry semantic meaning for humans, they are treated only as observation dimensions by agents.

\subsubsection{Behavior Encoding via Action-Sequence Transformation.}
For each action in the sliding window, it can be encoded into $HV_a$ through discretization and sampling in the hyperdimensional space.
To model temporal dependencies in action sequences, we define a time dictionary following a construction similar to the numeric dictionary (Equ.\ref{number}). 
Unlike perception encoding, actions are self-related and thus do not require identity binding and entity-wise bundling. The transformation can be represented as:
\begin{equation}
    HV_A =  \sum_{t'=t-w}^{t-1} {\oplus(HV_{a_{t'}} \otimes HV_{t'})}. 
\label{action_hd}
\end{equation}

\subsection{Receiver-Centric Fusion for Conversion Fidelity}
Standard HDC can only transmit static sender-centric semantic information, but has no control over how this information is interpreted to serve receivers' downstream decisions. This limitation directly breaks the conversion fidelity, inspiring our designs for semantic alignment and preservation.
\subsubsection{Ego-Binding for Semantic Alignment}
To align the broadcasted messages to the receivers' local context, Herm introduces the Ego-binding mechanism prior to message fusion. Receiver $i$ structurally binds the incoming message with its own encoding $m_i$ (the ego vector) using the circular convolution operator: $m_{i,j}' = m_j \otimes m_i$.

While the binding operation is a native primitive of HDC, existing works only use it for concept composition. Here, we utilize it for semantic contextualization. In the HRR framework, the binding operation functions as a dense multiplicative interaction. By binding $m_j$ and $m_i$, the mechanism inherently constructs a joint state representation at the feature level. This extracts second-order correlations between the sender's and receiver's semantics to support the downstream decision.
Moreover, the ego-binding preserves the relative semantic topology between diverse messages, ensuring that $\cos(m_j \otimes m_i, m_k \otimes m_i) \approx \cos(m_j, m_k)$. This property preserves algebraic consistency between semantic contextualization and the subsequent fusion. The verification is in Appendix \ref{Similarity-Preservation}.

\subsubsection{Similarity-Aware Message Fusion for Semantic Preservation}
Received messages inevitably contain redundant information, and indiscriminate fusion will lead to the dilution of decision-valid messages. 
In contrast to neural vectors whose interrelationships can be captured through attention mechanisms \cite{TGCNET}, evidence theory \cite{T2MAC}, or mutual information \cite{MAIC}, the correlation between hyperdimensional vectors is predominantly characterized by similarity measures.

We contend that similarity between observations serves as a crucial modulator in multi-agent communication fusion, with downstream effects on decision-making. 
\textbf{For supplementing environment perception}, dissimilar observations carry richer additional information and should be prioritized; \textbf{for understanding teammate behavior}, similarity in observations reflects analogous decision contexts that promote understanding and coordination between agents.

Hence, we propose a dual-branch message fusion mechanism that decouples observations from actions, processing each through distinct fusion pathways.
Agent $i$ first computes similarity scores as:
\begin{equation}
    \boldsymbol{\alpha}_i=\text{softmax}\{\alpha_1, \cdots, \alpha_i, \cdots, \alpha_N\},
\end{equation}
where $\alpha_{j \in [1,N]}$ denotes the cosine similarity between its own observational encoding $m_i^o$ and the received ones $m_{j}^o$.
Then, the similarity scores are inverted, resulting in the dissimilarity scores:
\begin{equation}
    \boldsymbol{\dot\alpha}_i=\text{softmax}\{-\alpha_1, \cdots, -\alpha_i, \cdots, -\alpha_N\}.
\end{equation}
For the perception branch, we employ dissimilarity scores $\boldsymbol{\dot\alpha}$ for fusion, whereas similarity scores $\boldsymbol{\alpha}$ are utilized for the action branch. The fusion is expressed as Equ.\ref{fusion}.
The fused messages $\mathcal{M}_i^o$ and $\mathcal{M}_i^a$ undergo a mapping transformation into neural vectors $c_i^o$ and $c_i^a$, which are then aggregated to produce the final message representation $\boldsymbol{c}_i$.
The overall training architecture adopts the value decomposition framework, where the message representation $\boldsymbol{c}_i$ and the agent's own trajectory encoding $h_i$ are fed into the local Q-network for decision-making. The pseudocode for training is in Appendix \ref{Pseudocode}.
\begin{equation}
        \mathcal{M}_i^o = \sum_{j}^{N} \boldsymbol{\dot\alpha_{i,j}}*(m_i^o \otimes m_j^o), \quad
        \mathcal{M}_i^a = \sum_{j}^{N} \boldsymbol{\alpha_{i,j}}*(m_i^a \otimes m_j^a).
\label{fusion}
\end{equation}

\subsection{Iterative Message Decoding via Similarity}
\label{section-decode}
For the explicit conversion of HSM, we devise a transformation from messages to semantics.
The process makes full use of the following properties of bind ($\otimes$), unbind ($\circledast$), and bundle ($\oplus$): 
 (1) $\delta(HV_p \oplus HV_q, HV_p) \gg 0$, (2) $\delta(HV_p \circledast HV_q, HV_p) \approx 0$, and (3) $\delta(HV_p \otimes HV_q \circledast HV_q, HV_p) \gg 0$, where $\delta$ denotes the similarity function.
Here, we specifically refer to the calculation of Holographic Reduced Representation. Proof is illustrated in Appendix \ref{proof}.



Since the perception encoding is obtained through entity-wise bundling in Equ.\ref{obs_hd}, Property 1 implies that $HV_{O}$ maintains similarity with each individual entity encoding $HV_{O}^{j}$.
When extracting information about a specific entity from observations, we perform (i) identity unbinding $HV_{O} \circledast HV_{id}$.
With Properties 2 and 3, $HV_{O}^{j}$ is successfully recovered while $HV_{O}^{-j}$ is disrupted.
Similarly, since $HV_{O}^{j}$ is calculated by attribute-wise bundling, extracting the value of a specific attribute requires performing (ii) attribute unbinding $HV_{O}^j \circledast HV_{attr}$.
Then, we compute the similarity between the resulting $HV_v$ and the value dictionary $\mathcal{V}$ to obtain the final decoded value $v$. The whole process is represented as: 
\begin{equation}
    v=\argmax_{v \in \mathcal{V}} \delta(HV_v, HV_O \circledast HV_{id} \circledast HV_{attr}).
\label{decode_obs}
\end{equation}
Analogously to Equ.\ref{decode_obs}, we can derive actions at any timestep $t' \in [t-w, t-1]$ from action-sequence encoding $HV_A$ through the unbinding operation of a specific timestep and the similarity matching:
\begin{equation}
    a_{t'}=\argmax_{a \in \mathcal{A}} \delta(HV_a, HV_A \circledast HV_{t'}).
\label{decode_action}
\end{equation}
Notably, the decoded semantics can be converted into natural language, providing a novel insight into the integration of multi-agent communication and large language models (LLMs).
Previous work \cite{li2024language} attempts to map agents' messages to the LLM embedding space for implicit collaboration with LLMs, which is constrained by the task complexity and the foundation model.
By contrast, we leverage HSM as an explicit bridging medium between agents and LLMs.
It establishes a conversion pathway that spans message encoding, symbolic semantics, and natural language.
We elaborate on the RL-LLM hybrid decision-making and corresponding experiments in Appendix \ref{llm+comm}.

\section{Experiments}
To validate the proposed Herm, we conduct extensive experiments on standard MARL benchmarks, including SMAC \cite{SMAC}, SMACv2 \cite{SMACV2}, Lbforaging \cite{LBF}, and Hallway \cite{NDQ}.
As baselines, we adopt MAIC \cite{MAIC}, T2MAC \cite{T2MAC}, CoDe \cite{CODE}, and TGCNET \cite{TGCNET}. 
For transmission fidelity, we compare CroMAC \cite{CroMAC2024} and several enhanced methods that integrate the AME \cite{AME2023} mechanism as robustness-oriented baselines.
Then, we conduct semantic visualization to verify the conversion fidelity. 

For fairness, we adopt a full communication topology for all methods across all benchmarks. We clarify that the attribute dictionary only provides a prior for explicit semantic conversion, and agents are completely agnostic to attribute semantics, thus bringing no performance gain.

\subsection{Performance in Noise-Free Environments}

\begin{figure*}[t]
\centering
\includegraphics[width=1.0\linewidth]{ 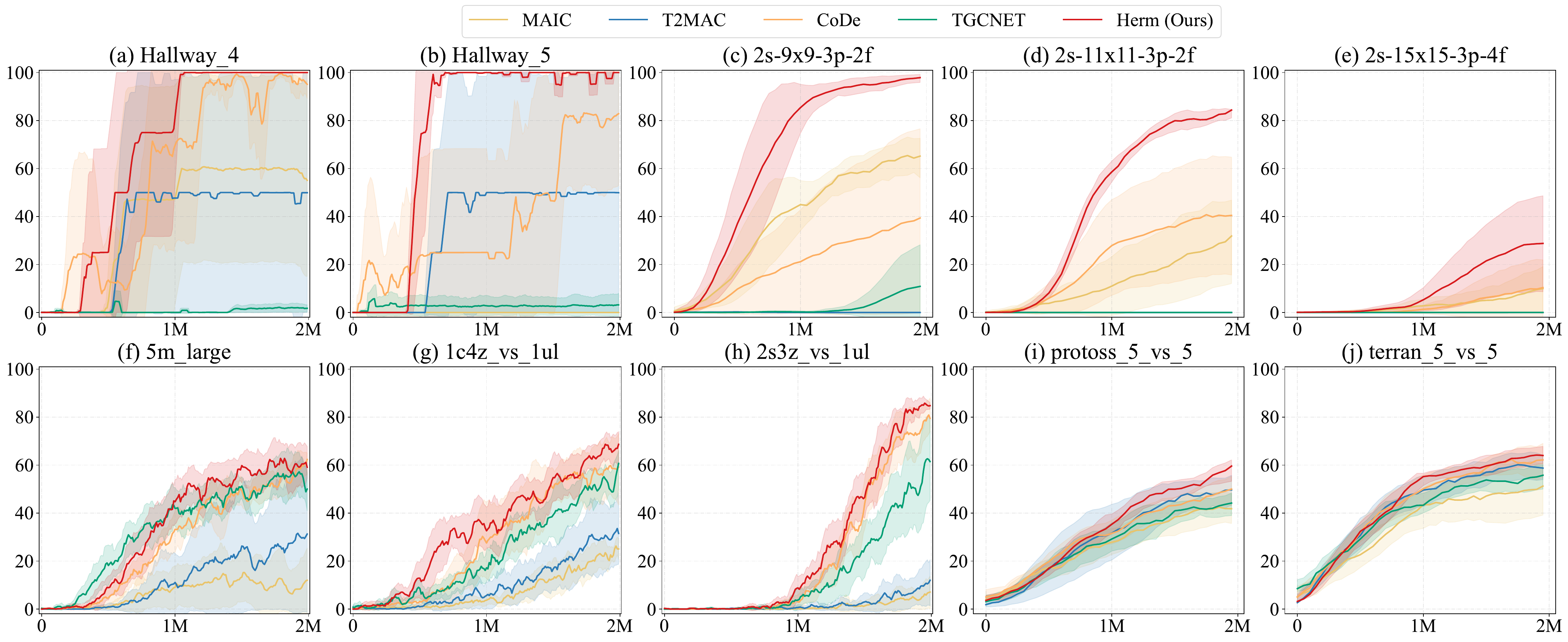}
\caption{Performance in Noise-Free Environments. Each curve represents the average across four random seeds. (a)-(b) Hallway. (c)-(e) LBF . (f)-(h) SMAC. (i)-(j) SMACv2.}
\label{no-noise-result}
\end{figure*}

We first verify the efficacy of Herm for multi-agent collaboration.
In Fig. \ref{no-noise-result}, Herm consistently outperforms all baselines across four benchmarks in noise-free settings.
 For the Hallway task with a low-dimensional observation space, the orthogonal HSM encoding enables more explicit and efficient transmission of agents' decision contexts, leading to faster convergence and a higher task completion rate.
 For the partially observable LBF task, Herm achieves remarkable performance gains, especially in large-scale scenarios where baselines almost fail to converge.
 This is attributed to our similarity-aware dual-branch fusion mechanism, which precisely matches the task's core requirements of unobserved information supplementation and adjacent agent coordination.  
 For communication-intensive SMAC and SMACv2 tasks with high randomness, Herm addresses two core limitations of baselines: 1) numerical messages \cite{MAIC,T2MAC} can only deliver behavioral recommendations, failing to support early-stage joint exploration; 2) neural messages \cite{CODE,TGCNET} rely on black-box fusion and dilute observation variations caused by environmental stochasticity. In contrast, the binding operation in HSM amplifies discriminative observation features, and our interpretable fusion design provides explicit guidance for both exploration and coordination, thus achieving superior performance.

\subsection{Performance in Noisy Environments}
To verify the transmission fidelity, we conduct the robustness evaluation in two noisy settings.

\textbf{Random Noise.} In this setting, each inter-agent channel is subject to independent and identically distributed perturbations. 
We adopt the noise configuration from MAGI \cite{MAGI} as:
\begin{equation}
   \hat{m}_i = m_i + \eta_r \cdot \phi \cdot \epsilon,
\label{noise}
\end{equation}
$\epsilon$ represents Gaussian noise, $\phi$ denotes amplitude computed as the mean of maximum values across all messages, and $\eta_r \in \{0,0.5,1,1.5,2\}$ is the scaling weight of noise.
In Fig.\ref{noisy_test}, Herm maintains superior robustness across varying noise, specifically in the linear (Hallway) or grid-based (LBF) scenarios. 
We attribute this to the well-preserved orthogonality in symbolic message space, where semantic and similarity relationships remain stable despite noise interference.
Although the AME method can somewhat enhance the robustness of baselines, it is not suitable for small-scale scenarios such as Hallway, where each message passing remains essential.
While CroMAC attempts to improve message-level robustness, its reliance on neural representations limits its noise resilience.

\begin{figure*}[ht]
\centering
\includegraphics[width=1.0\linewidth]{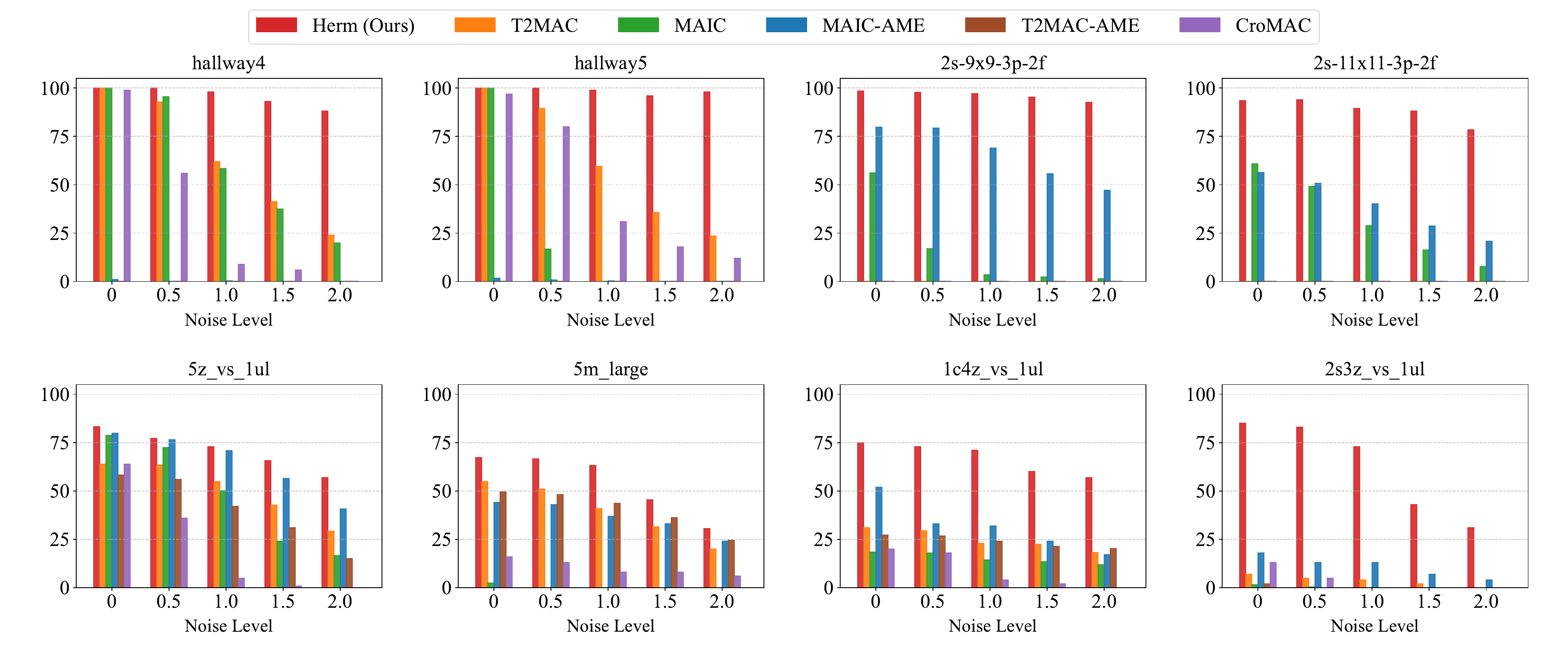}
\caption{Robustness Evaluation in Randomly Noisy Environments.}
\label{noisy_test}
\end{figure*}

\textbf{Adversarial Noise.} We utilize two types of gradient-based adversarial attacks.
(1) Fast Gradient Sign Method (FGSM) \cite{FGSM} is a one-step attack as Equ.\ref{FGSM}. $\eta_f \in \{0.3, 0.5\}$ is the scaling parameter and $\phi$ denotes the amplitude used in Equ.\ref{noise}. $f$ maximizes the cross-entropy loss between the noisy policy and the noise-free policy.
\begin{equation}
\label{FGSM}
    \hat{m}_i = m_i + \eta_f \cdot sign(\nabla_{m_i}f(m_i)) \cdot \phi,
\end{equation}
(2) Projected Gradient Descent (PGD) \cite{PGD} is a multi-step variant of FGSM, defined as 
\begin{equation}
    \hat{m}_i^{e+1} = \Pr_{\phi} \left(\hat{m}_i^{e} + \eta_f \cdot sign(\nabla_{\hat{m}_i^{e}}f(\hat{m}_i^{e})) \cdot \frac{\phi}{5}\right), 
\end{equation}
where $e \in [0,4]$ means 5-step attacks and $\hat{m}_i^{0} = m_i$. Here, we adopt the L2 norm ball constraint $\Pr$.

Results in Tab.\ref{tab:adv_robust_evaluate} indicate that Herm almost achieves optimal robustness in two settings. Under PGD attacks, Herm exhibits slight performance fluctuations.
This is because the entire communication process only involves a mapping network.
This structure enables unobstructed backpropagation of adversarial gradients to original messages, thereby increasing the risk of communication loss. 

\begin{table*}[t]
    \centering
    \scriptsize
    \caption{Robustness Evaluation in Adversarially Noisy Environments. }
    \label{tab:adv_robust_evaluate}
    
    \begin{tabular}{cccccccc}
        \toprule
        \textbf{Scenario} & \textbf{Method} & \textbf{No-Adv} & \textbf{FGSM (0.1)} & \textbf{FGSM (0.3)} & \textbf{FGSM (0.5)} & \textbf{PGD (5, 0.3)} & \textbf{PGD (5, 0.5)} \\
        \midrule
        \multirow{3}{*}{\textbf{5m\_large}} & CroMAC & $16.3 \pm 0.2$ & $10.7 \pm 0.1$ & $3.1 \pm 0.4$ & $0.5 \pm 0.9$ & $9.2 \pm 0.3$ & $1.6 \pm 0.7$ \\

        & MAIC-AME & $48.9 \pm 0.8$ & $44.1 \pm 0.6$ & $26.5 \pm 0.5$ & $15.1 \pm 0.2$ & $42 \pm 3.3$ & $\mathbf{36 \pm 2.8}$ \\

        & Herm (Ours) & $\mathbf{67.2 \pm 0.5}$ & $\mathbf{58.3 \pm 0.4}$ & $\mathbf{31.8 \pm 0.1}$ & $\mathbf{18.6 \pm 0.3}$ & $\mathbf{43.5 \pm 1.4}$ & $33.9 \pm 3.2$ \\
        \hline
        \multirow{3}{*}{\textbf{5z\_vs\_1ul}} & CroMAC & $64.5 \pm 0.1$ & $10.5 \pm 0.8$ & $0.3 \pm 0.2$ & $0.8 \pm 0.5$ & $31.4 \pm 0.9$ & $24.7 \pm 0.1$ \\

        & MAIC-AME & $77.8 \pm 0.3$ & $60.8 \pm 0.7$ & $34.2 \pm 0.6$ & $6.2 \pm 0.1$ & $\mathbf{65.2 \pm 2.4}$ & $37.5 \pm 3.3$ \\

        & Herm (Ours) & $\mathbf{79.0 \pm 0.2}$ & $\mathbf{68.5 \pm 0.9}$ & $\mathbf{35.5 \pm 0.3}$ & $\mathbf{25.1 \pm 0.4}$ & $63.5 \pm 4.1$ & $\mathbf{58.5 \pm 3.7}$ \\
        \hline
        \multirow{3}{*}{\textbf{Hallway 4}} & CroMAC & $99.3 \pm 0.6$ & $16.2 \pm 0.5$ & $6.0 \pm 0.1$ & $0.0 \pm 0.8$ & $39.5 \pm 0.7$ & $17.3 \pm 0.2$ \\

        & MAIC-AME & $5.3 \pm 0.4$ & $0.0 \pm 0.3$ & $0.0 \pm 0.6$ & $0.0 \pm 0.0$ & $0.0 \pm 0.0$ & $0.0 \pm 0.0$ \\

        & Herm (Ours) & $\mathbf{100 \pm 0.0}$ & $\mathbf{99.0 \pm 0.2}$ & $\mathbf{95.5 \pm 0.4}$ & $\mathbf{92.7 \pm 0.5}$ & $\mathbf{100.0 \pm 0.0}$ & $\mathbf{97.5 \pm 1.6}$ \\
        \hline
        \multirow{3}{*}{\textbf{Hallway 5}} & CroMAC & $98.1 \pm 0.9$ & $96.9 \pm 0.7$ & $24.6 \pm 0.3$ & $13.5 \pm 0.2$ & $70.8 \pm 0.4$ & $44.1 \pm 0.6$ \\

        & MAIC-AME & $3.6 \pm 0.5$ & $0.0 \pm 0.0$ & $0.0 \pm 0.9$ & $0.0 \pm 0.0$ & $0.0 \pm 0.0$ & $0.0 \pm 0.0$ \\

        & Herm (Ours) & $\mathbf{100.0 \pm 0.0}$ & $\mathbf{98.0 \pm 0.6}$ & $\mathbf{93.5 \pm 0.8}$ & $\mathbf{90.2 \pm 0.7}$ & $\mathbf{100.0 \pm 0.0}$ & $\mathbf{100.0 \pm 0.0}$ \\
        \bottomrule
    \end{tabular}
\end{table*}
\begin{figure*}[t]
\centering
\includegraphics[width=0.85\linewidth]{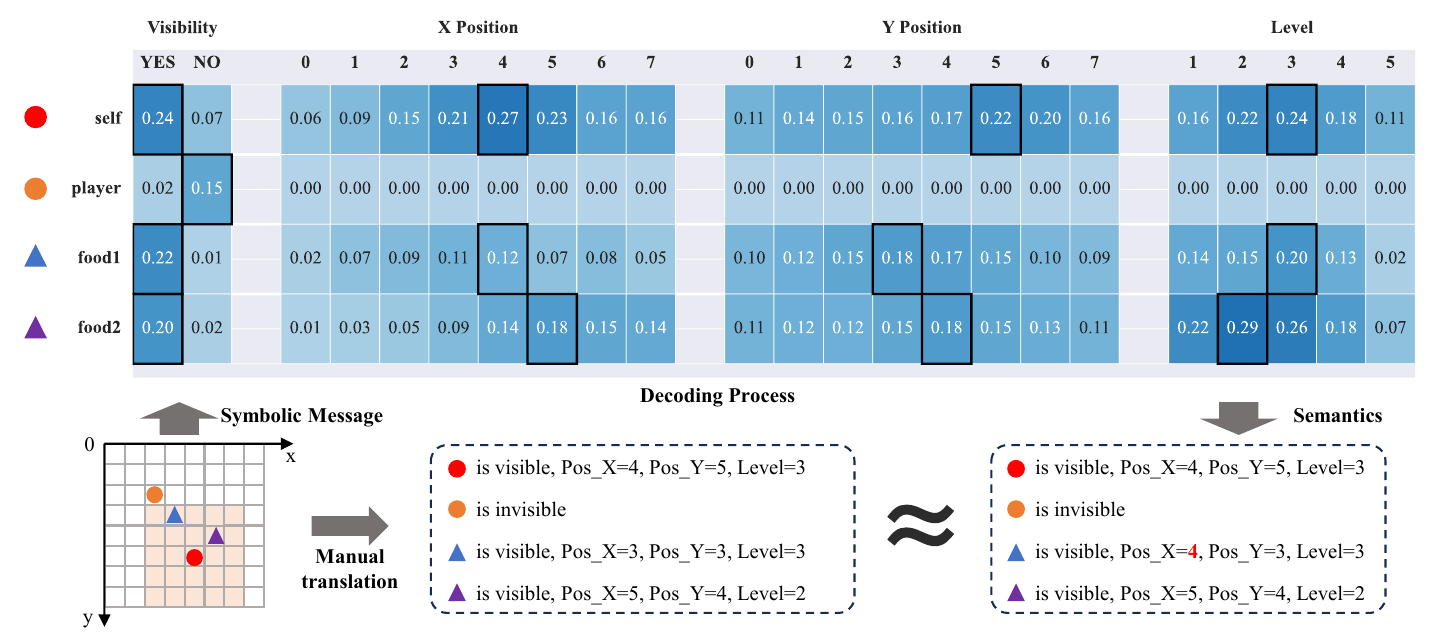}
\caption{Message-Semantic Transformation in LBF. The upper displays entity-wise similarity measurements, while the lower shows: the raw observations, manually translated observation texts, and the decoded semantics. The black box highlights the value exhibiting maximum similarity.}
\label{explain}
\end{figure*}

\subsection{Message-Semantics Transformation}
The explicit conversion fidelity lies in the structural readability, which is guaranteed by the algebraic properties of hyperdimensional computing, and is completely independent of manual semantic labeling.
Even without any labels, we can still accurately disassemble the composite message, clarify its information composition, and measure the similarity between different components. Manual labels are only an optional human-readable interface, which maps the disassembled message components to natural language for human inspection, and do not constitute the source of readability itself.

Fig.\ref{explain} presents the decoding results for the observation components in messages within the LBF.
The decoded semantics obtained through attribute-value combinations show near-perfect alignment with manually translated observation texts.
Such alignment verifies that hyperdimensional symbolic messages possess excellent conversion fidelity.
The conversion for decision-valid semantics is then verified via subsequent ablation experiments.
Furthermore, we observe that the accuracy of semantic transformation improves with the expansion of the hyperdimensional space. 
More semantic transformations and LLM collaboration are presented in  Appendix \ref{Interpretability-Analysis},\ref{Similarity-Analysis}, and \ref{llm+comm}.

\begin{figure}[htbp]
  \centering
  \subfloat[Analysis of Communication Components in MMM2.]{\label{fig:ablaa}
    \includegraphics[width=0.57\textwidth]{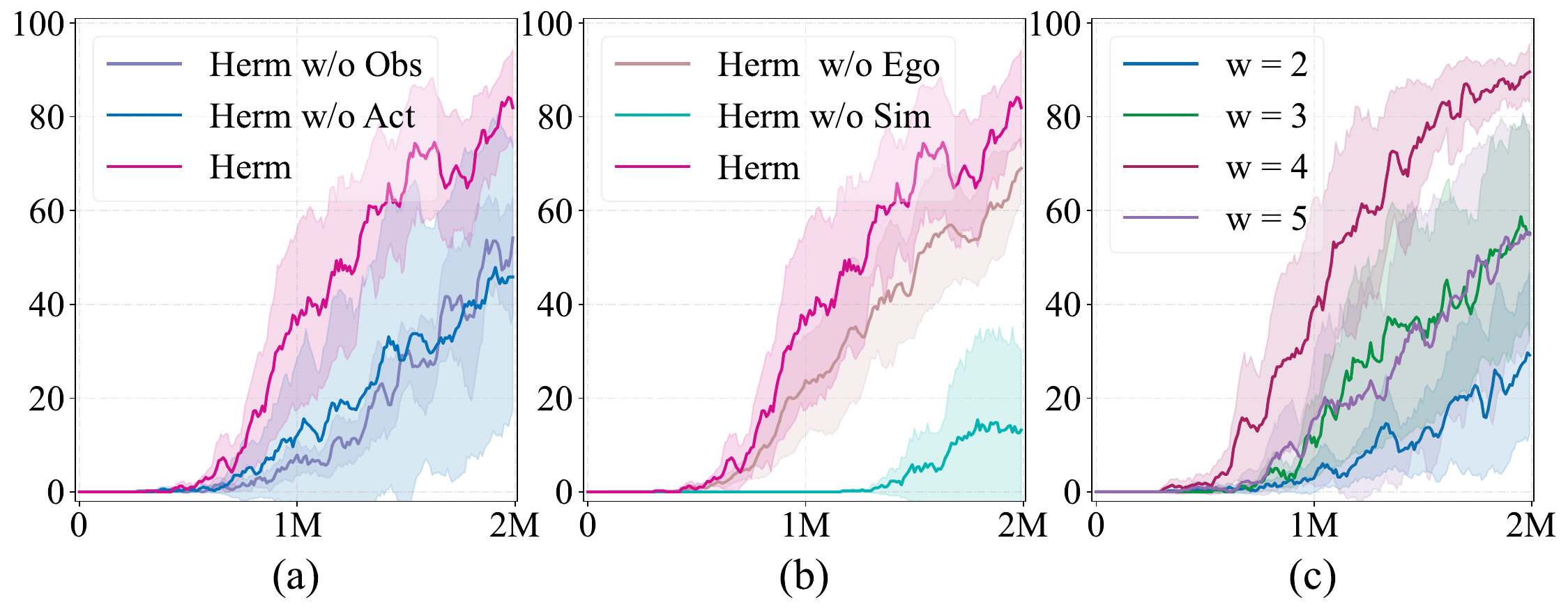}
  }
  \hfill 
  \subfloat[Performance with Varying $\mathcal{H}$ in LBF.]{\label{fig:ablab}
    \includegraphics[width=0.38\textwidth]{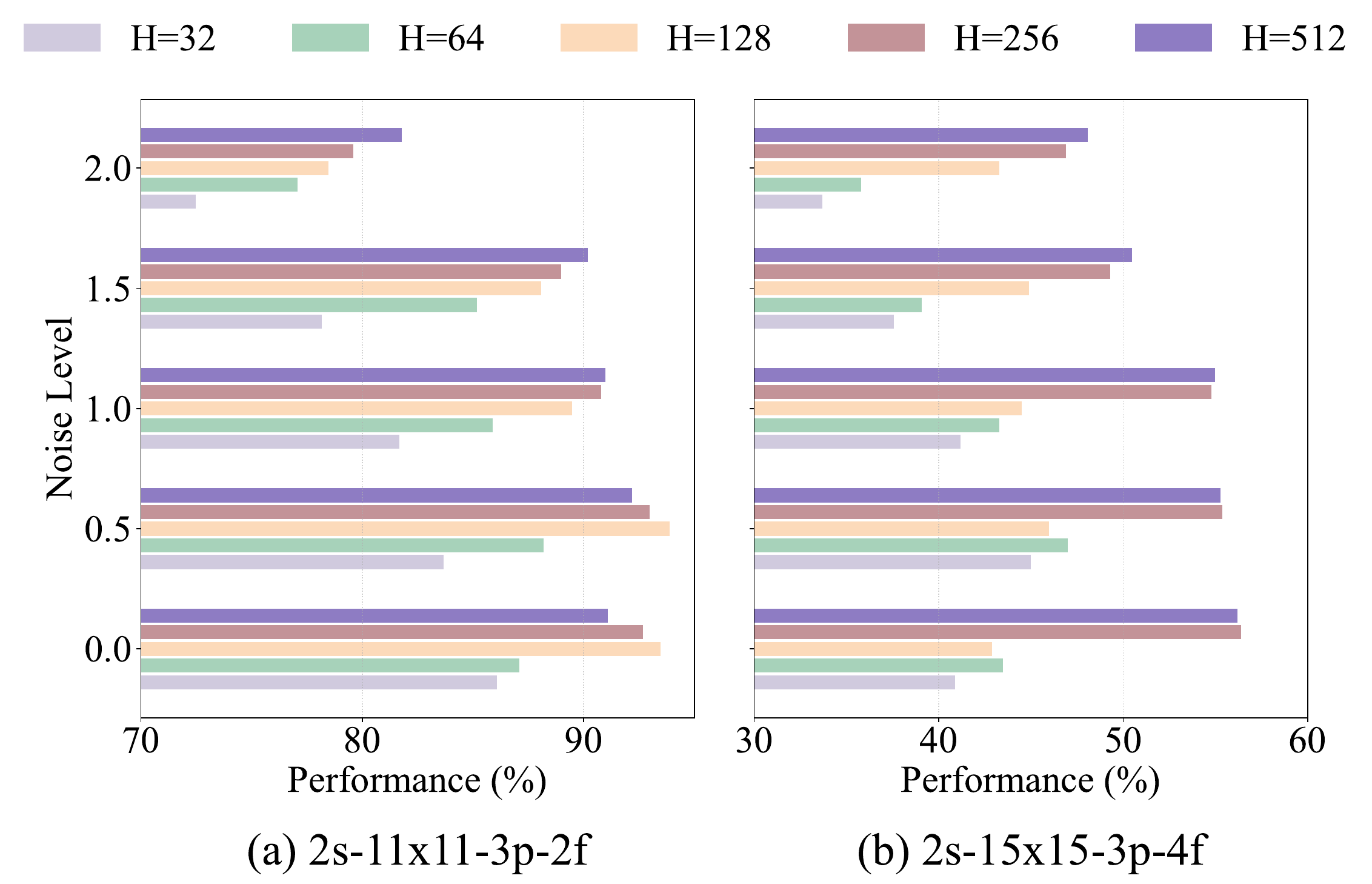}
  }
  \caption{Ablation and Parameter Analysis on the Communication Process.}
  \label{fig:main}
\end{figure}



\subsection{Ablation over Communication Process}
We first analyze components in the sender's message encoding, including 1) observation and 2) action sequence.
Fig.\ref{fig:ablaa} (a) demonstrates the necessity of both components. 
Then, we conduct ablations in the receiver's message fusion. 
In Fig.\ref{fig:ablaa} (b), eliminating similarity metrics leads to catastrophic performance degradation. 
Massive messages result in information redundancy, making it impossible for agents to extract decision-valid semantics.
Similarity offers explicit fusion guidance in terms of perception supplementation and behavior coordination.
Moreover, without ego-binding, Herm suffers from semantic misalignment, which increases the burden on the downstream network to capture correlations. It further results in reduced sample efficiency and performance.

\subsection{Parameter Analysis}
\textbf{Action Sequence Length $w$.} Fig.\ref{fig:ablaa} (c) shows significant variations in behavioral intention transmission with different action sequence lengths. When $w=2$, limited action sequences hinder effective inter-agent behavioral understanding. Performance shows marked improvement with increasing $w$. However, excessively large $w$ can compromise the orthogonality in message encoding, consequently degrading collaborative performance. In our study, we adopt $w=4$ as the default setting.


\textbf{Message Dimension $\mathcal{H}$.}
The impact of $\mathcal{H}$ is quantified in Fig.\ref{fig:ablab}. When $\mathcal{H}=32$, Herm performs poorly in both noiseless and noisy settings. Distinct semantic symbols become entangled in the constrained hyperdimensional space, leading to the ambiguous perception of the environment and teammates. As $\mathcal{H}$ increases, Herm exhibits enhanced robustness. Moreover, the required dimensionality $\mathcal{H}$ varies in different scenarios. For simpler tasks (Fig.\ref{fig:ablab} (a)), $\mathcal{H}=128$ suffices to handle most conditions, whereas for complex tasks (Fig.\ref{fig:ablab} (b)), $\mathcal{H}=256$ exhibits significant advantages.
In fact, a larger $\mathcal{H}$ delivers better performance in robustness (Fig.\ref{fig:ablab}) and readability (see Appendix \ref{Interpretability-Analysis}).

\subsection{Transmission Fidelity Quantification}
We adopt the Spearman correlation coefficient ($\rho$) as the quantitative metric~\cite{spearman}.
The Spearman correlation coefficient is a nonparametric statistic that quantifies the strength and direction of monotonic relationships between two variables by evaluating their rank order.
We calculate the Spearman correlation between the similarity ranking of raw observations before encoding and the similarity ranking of encoded messages after transmission.
The quantitative results are summarized in Table~\ref{tab:rho_noise}.
A $\rho$ value closer to 1 indicates that the HSM module better preserves the original semantic similarity relationships.
Experimental results show that such similarity relationships can be stably maintained under various noise perturbations, which verifies the strong transmission fidelity of the HSM module.

\begin{table}[tbp]
  \centering
  \footnotesize
  \caption{Spearman Correlation Coefficient $\rho$ under Noise (mean $\pm$ std)}
  \label{tab:rho_noise}
  \begin{tabular}{llccccc}
    \toprule
    Task & Encoding & noise=0 & noise=0.5 & noise=1.0 & noise=1.5 & noise=2.0 \\
    \midrule
    $9\times9$ & HDC & $\textbf{0.878}\pm\textbf{0.017}$ &  $\textbf{0.872}\pm\textbf{0.020}$ & $\textbf{0.785}\pm\textbf{0.034}$ & $\textbf{0.678}\pm\textbf{0.047}$ & $\textbf{0.615}\pm\textbf{0.050}$ \\
    $9\times9$ & Neural & $0.862\pm0.039$ &  $0.825\pm0.045$ & $0.720\pm0.075$ & $0.595\pm0.070$ & $0.440\pm0.048$ \\
    $11\times11$ & HDC & $\textbf{0.882}\pm\textbf{0.028}$  & $\textbf{0.875}\pm\textbf{0.020}$ & $\textbf{0.785}\pm\textbf{0.028}$ & $\textbf{0.705}\pm\textbf{0.030}$ & $\textbf{0.592}\pm\textbf{0.038}$ \\
    $11\times11$ & Neural & $0.858\pm0.040$  & $0.815\pm0.022$ & $0.695\pm0.072$ & $0.510\pm0.118$ & $0.385\pm0.082$ \\
    \bottomrule
  \end{tabular}
\end{table}

\section{Conclusions}
This paper pioneers the definition of semantic fidelity in multi-agent communication. Based on this, the hyperdimensional symbolic message and its communication framework are proposed to enhance both transmission fidelity and conversion fidelity.
However, this paper still has certain limitations to be addressed. For one thing, the encoding module is decoupled from the policy optimization; future work can explore an end-to-end joint training paradigm from semantic representation to decision-making. 
For another, the per-dimension message encoding scheme lacks flexibility. We will therefore design an automatic fine-grained attribute selection and dynamic encoding mechanism, which allows agents to decide what to communicate.

\section{Acknowledgments}
This work was supported by the National Natural Science Foundation of China (Grant No. 62576029) and Beijing Natural Science Foundation (Grant No. L2607023).
\medskip

{
\small
\bibliographystyle{unsrtnat}
\bibliography{neurips_2026}
}


\appendix

\section{Supplementary Experiments}
\subsection{SMAC Configuration}
In the experiment, we utilize the communication-intensive tasks proposed in \cite{CODE}.
As illustrated in Fig. \ref{map_img}, the initial positions of units are randomly generated with specific terrain obstacles. Agents are required to explore their surrounding environment to locate allies and enemies. Given that both sides boast comparable military strength and enemies exhibit a tendency to cluster, solo agents are vulnerable to elimination, ultimately resulting in mission failure due to force imbalances. Therefore, agents must wait to form clusters before engaging in combat under conditions of force balance; otherwise, successful mission completion becomes exceedingly difficult. Hence, the imperative of inter-agent communication is further underscored, accompanied by heightened requirements for exploration and coordination. This renders these tasks highly suitable for assessing communication-aware multi-agent reinforcement learning algorithms.
\begin{figure*}[ht]
    \centering
    \includegraphics[width=0.9\linewidth]{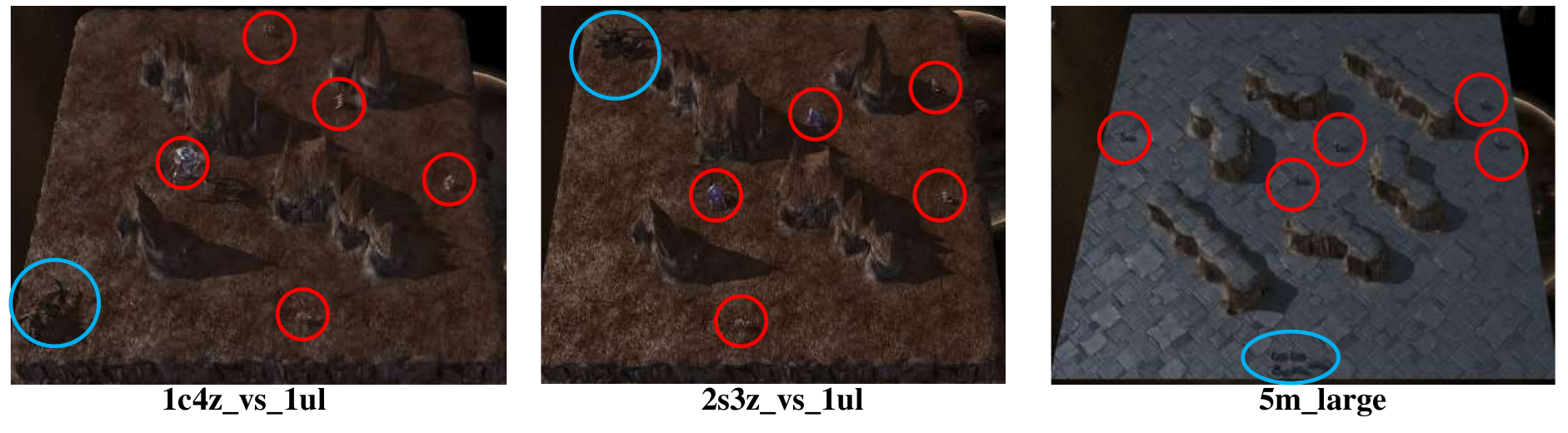}
    \caption{Communication-Intensive Tasks in SMAC \cite{CODE}. Enemies are highlighted in blue circles, while allies are marked in red.}
    \label{map_img}
\end{figure*}
\subsection{Experiment Parameter}
In the experiments, both the baselines and our proposed approach are implemented using the PyMARLzoo+ architecture \footnote{https://github.com/ailabdsunipi/pymarlzooplus} \cite{pymarlzooplus}, retaining most default parameters (see Tab.\ref{tab:default-para}). In the experiment, the CPU used is AMD Ryzen 9 9950X, and the GPU is NVIDIA GeForce RTX 4090.

\begin{table}[ht]
    \centering
    \caption{Default Parameters}
    \begin{tabular}{cc}
    \toprule
       Parameter  & Value \\
    \midrule
       Optimizer  & Adam \\
       Learning Rate & $5 \times 10^{-4}$ \\
       Batch Size & 32 \\
       Discount Factor & 0.99 \\
       GRU Hidden State Dim & 64 \\
       Mixing Hidden Dim & 32 \\
       Hypernet Embedding Dim & 32\\
    \bottomrule
    \end{tabular}
    \label{tab:default-para}
\end{table}

Table \ref{tab:parameters} presents the parameters specific to our proposed Herm.
Due to the simplicity of the hallway task, we uniformly employ $\mathcal{H}=64$ and $w=1$ configurations. For most LBF tasks, we maintain $\mathcal{H}=64$ and $w=1$, while increasing to $H=256$ for more complex 15×15 maps. For SMAC with its greater entity and attribute complexity, we adopt $\mathcal{H}=256$ with $w=3$ or 4 to adequately represent the expansive observation space.

\begin{table}[htbp]
\centering
\caption{Method-Specific Parameters}
\label{tab:parameters}
\begin{tabular}{p{1cm}p{7.2cm}p{2cm}}
\toprule
\textbf{Symbol} & \textbf{Description} & \textbf{Value} \\
\midrule
$\mathcal{H}$ & Dimension of the hyperdimensional space & $\{64, 128, 256\}$ \\
$w$ & Length of historical action sequences & $\{1, 3, 4\}$ \\
$Bin$ & Number of discretized partitions & $\{5, 10\}$ \\
\bottomrule
\end{tabular}
\end{table}

To establish a robust communication baseline, we implement the message ablation mechanism from AME, leveraging multi-round voting to mitigate noise. Specifically, during training, we retain $K$ out of $N-1$ message sources for policy learning, where $K$ is set to $N-2$ or $N-3$. During testing, we perform $P=10$ ablation trials for stabilized evaluation.

\subsection{Analysis of Model Size and Efficiency}
We compare the model parameter counts and the time required for a single complete interaction across different algorithms in Tab.\ref{tab:para-eff}. 
Results demonstrate that Herm benefits from hyperdimensional symbolic messages and eliminates the need for neural networks in the encoding and fusion processes, resulting in a minimal model parameter count. 
Herm’s parameters consist of two components: the first encompasses the GRU and MLP networks for processing agents’ local observation, and the second corresponds to the single-layer mapping network after message fusion.
In contrast, other baseline methods typically require additional message encoding networks, while adopting attention mechanisms and graph-based aggregation schemes for message fusion, which entails a considerably larger parameter overhead.

In terms of decision-making efficiency, Herm is limited by the programmatic encoding of hyperdimensional computing, exhibiting a slight disadvantage. However, the static encoding process allows for the reuse of previous results, trading space for time—similar to the KV cache mechanism. Results show that storing and reusing the binding outputs of attribute-value pairs can significantly reduce encoding time, bringing it close to baselines. 
\begin{table}[ht]
    \centering
    \caption{Comparison on Model Size and Efficiency}
    \begin{tabular}{ccccccc}
    \toprule
         & MAIC & T2MAC & CoDe & TGCNET& Herm (Ours) & Herm + Cache\\
    \midrule
       Params (K) & 50.79 & 50.70 & 49.39 & 827.01 & \textbf{46.88} & \textbf{46.88}\\
       Times (s)  & 1.481  & 1.528  & \textbf{1.363} & 1.656 &  1.993 & 1.392\\
    \bottomrule
    \end{tabular}
    
    \label{tab:para-eff}
\end{table}

\subsection{Additional Comparison on Robustness}
We further supplement additional robustness tests and incorporate the baseline, MAGI \cite{MDGI2024}. Since MAGI is not open-sourced, some results are missing. Results in Fig.\ref{fig:mdgi_compare} demonstrate that our proposed Herm still achieves the optimal robust performance in scenarios with a large number of entities.
\begin{figure*}[!htb]
     \centering
     \includegraphics[width=0.7\textwidth]{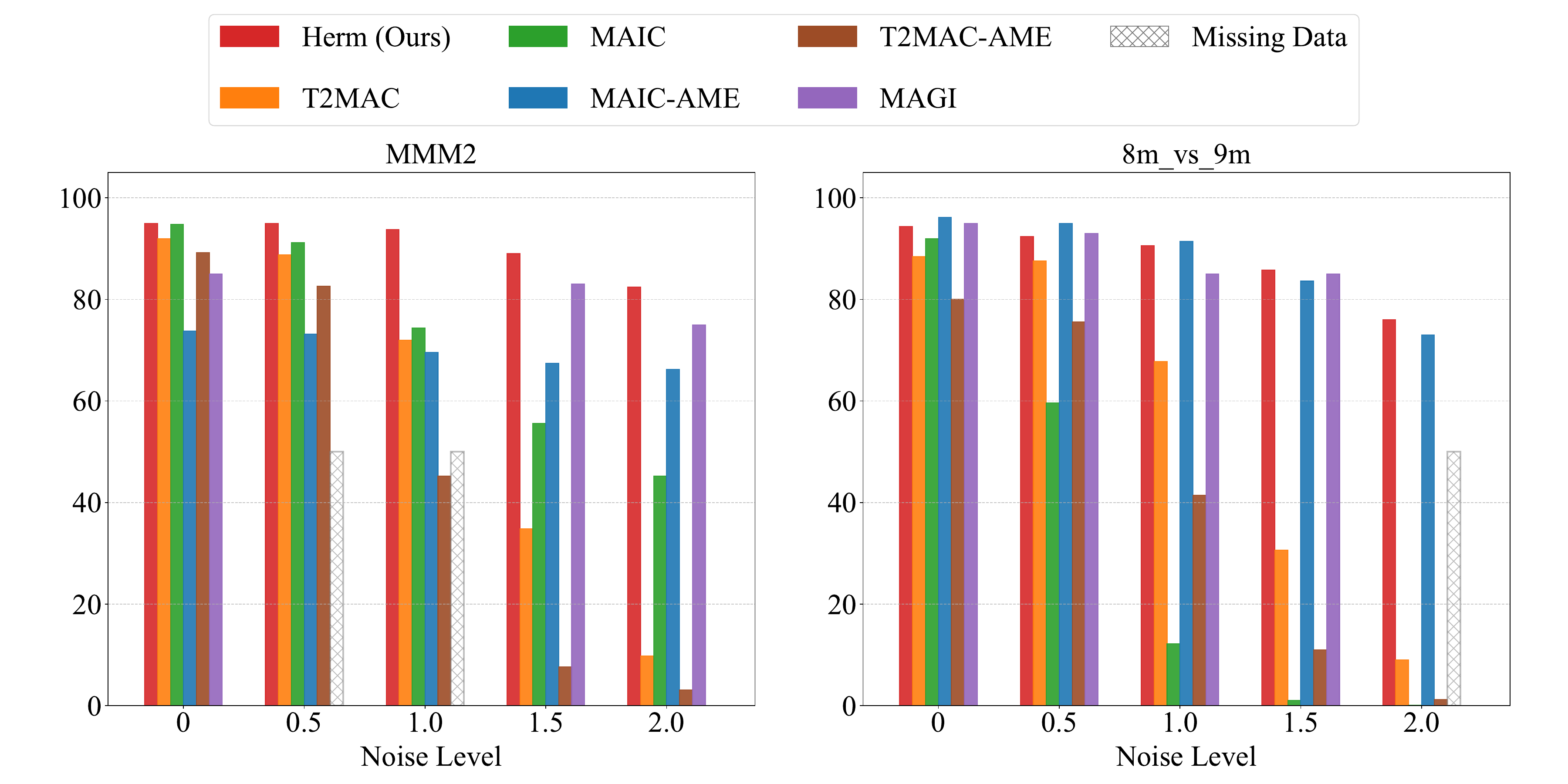}
     \caption{Additional Robustness Comparison.}
     \label{fig:mdgi_compare}
\end{figure*}

\subsection{Scalability of Hyperdimensional Computing Framework}
\label{append-hdc}
Different HDC frameworks differ primarily in their implementations of symbolic operations. We adopt the most conventional HRR framework in our experiments. 
We provide supplementary performance comparisons of different computational frameworks in Fig.\ref{fig:hd_type}.

 \begin{figure}[htb]
     \centering
     \includegraphics[width=0.65\linewidth]{ 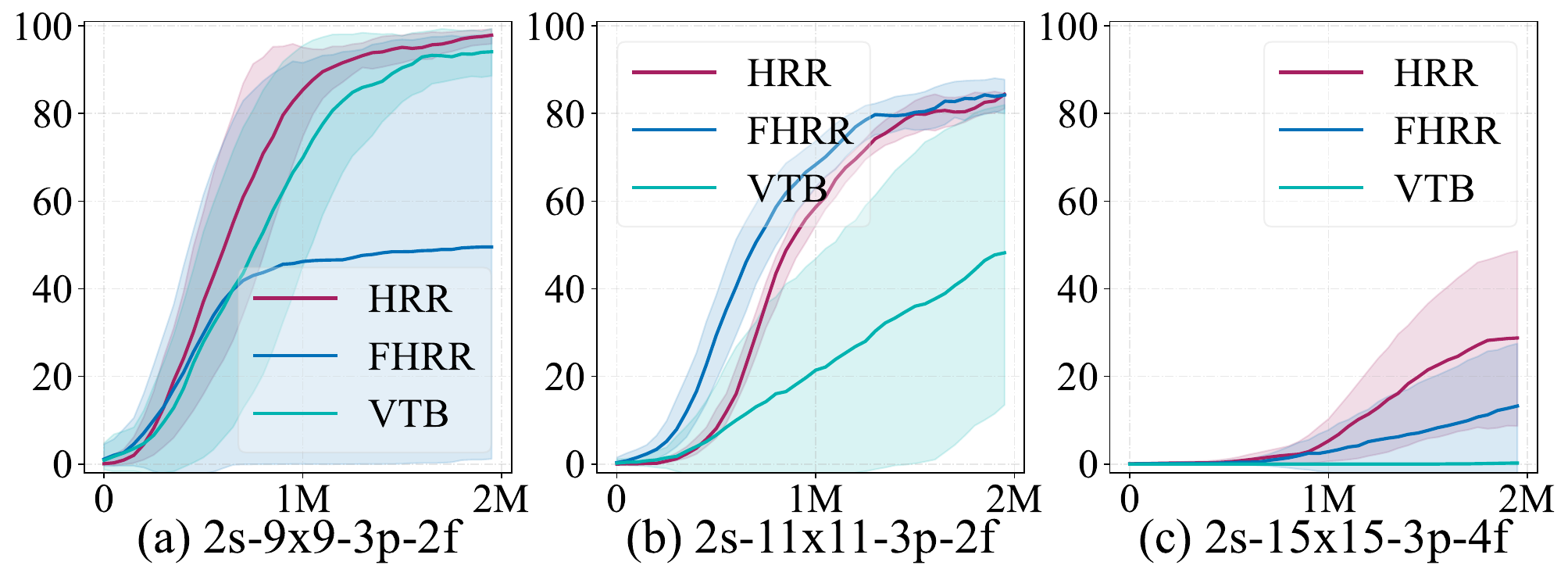}
     \caption{Performance Comparison with Different HDC Framework.}
     \label{fig:hd_type}
 \end{figure}

\textbf{(1) Holographic Reduced Representation (HRR)} \cite{HRR} is a technique in hyperdimensional computing that represents compositional structures using circular convolution in distributed representations, allowing for the binding and unbinding of vectors to combine concepts and retrieve them given the combined result. Its key operations involve circular convolution (binding) and circular correlation (unbinding): for vectors $\mathbf{A} $ and $\mathbf{B}$, binding is calculated as: 
\begin{equation}
    (\mathbf{A} \otimes \mathbf{B})_k = \sum_{i=0}^{D-1} A_i B_{k-i} \bmod D,
\end{equation}
 where $k$ refers to the value of the k-th dimension in the output, and $d$ is the size of the vector space. The unbinding uses a similar process, formulated as:
 \begin{equation}
    (\mathbf{A} \circledast \mathbf{B})_k = \sum_{i=0}^{D-1} A_i B_{i+k} \bmod D.
\end{equation}
 HRR advantages include 1) representing numerous key-value pairs in a fixed-size vector, 2) efficient associative memory for tasks like natural language processing, and 3) preservation of semantic similarity. 
However, its commutative binding limits the representation of non-commutative relationships and complex hierarchical structures.

\textbf{(2) Fourier Holographic Reduced Representation (FHRR)} \cite{FHRR} extends HRR by performing operations in the frequency domain, with hypervectors of the form $h=(e^{i\theta_1},\ldots,e^{i\theta_d})$. Bundling and binding are done via element-wise addition and multiplication, respectively. FHRR leverages Fourier transform properties to simplify binding, making it a simple element-wise multiplication in the frequency domain for lower computational complexity $O(D\log D)$, while also supporting sparse vectors well. 

\textbf{(3) Vector-Derived Transformation Binding (VTB)} \cite{VTB} proposes a new binding operation that represents relationships as linear transformations in vector space. For vectors $\mathbf{A} $ and $\mathbf{B}$, binding is achieved through a linear transformation matrix W as $\mathbf{A}\otimes\mathbf{B} = \mathbf{W} \cdot (\mathbf{A} + \mathbf{B})$
 where $+$ denotes concatenation. It allows flexible vector binding with a learnable transformation matrix $W$, and supports encoding hierarchical structures such as tree-like or knowledge graph relationships. VTB has advantages on high flexibility in representing complex relationships, suitability for neural implementations due to learnable matrices, and minimal influence on vector length beneficial for neural implementation.  

 \subsection{Extension with Numerical Messages}
This section further investigates the integration schemes between hyperdimensional computing (HDC) and the numerical message.
Since HDC typically does not involve neural networks, we attempt to realize integration from the perspectives of post-processing. 
Notably, the design presented in this section is merely intended as an intuitive exploration and verification; more in-depth integration and design will be reserved for future research.

Numerical messages utilized in \cite{vbc,TMC,MAIC,T2MAC} are usually presented in the form of Q-tables.
The sender conveys action suggestions as messages to the receiver, which synthesizes all suggestions to make its own decision. 
This indicates that a numerical message is formally an array.
Therefore, we propose to use hyperdimensional symbolic messages (HSM) to represent this array, thereby enhancing its robustness against noise.

Herein, we propose a simple conversion method to transform an array into a hyperdimensional vector. First, attributes are defined as positive/negative signs and finite digit positions, and the value set covers positive/negative signs and integers 0–9. With this division, the sender adopts the designed hyperdimensional encoding pipeline (see Section \ref{section-encode}) for conversion and transmission, while the receiver recovers the array from the vector through iterative decoding (see Section \ref{section-decode}).

\begin{table}[h]
    \centering
    \footnotesize
    \caption{Performance with HSM-Enhancement}
    \begin{tabular}{ccccccc}
    \toprule
     \multirow{2}{*}{Task} &\multirow{2}{*}{$\mathcal{H}$} & \multicolumn{5}{c}{Noise Level} \\
     \cline{3-7}
     & & 0 & 0.5 & 1.0 & 1.5 & 2.0\\
    \midrule
    \multirow{4}{*}{Hallway-4} &MAIC & 100 $\pm$ 0.0& 95.6 $\pm$ 1.7&58.4 $\pm$ 1.9& 37.4 $\pm$ 3.3& 20.0 $\pm$ 3.9\\
     &MAIC-E ($\mathcal{H}=256$) & 98.4 $\pm$ 0.5  & 96.8 $\pm$ 0.8 & 94.5 $\pm$ 0.7  & 64.5 $\pm$ 2.6 & 44.6 $\pm$ 4.1\\
     &MAIC-E ($\mathcal{H}=512$) & 100.0 $\pm$ 0.0   & 100.0 $\pm$ 0.0 & 100.0 $\pm$ 0.0   & 100.0 $\pm$ 0.0 & 85.5 $\pm$ 1.3\\
     &MAIC-E ($\mathcal{H}=1024$) & 100.0 $\pm$ 0.0  & 100.0 $\pm$ 0.0  & 100.0 $\pm$ 0.0  & 100.0 $\pm$ 0.0  & 100.0 $\pm$ 0.0\\
     \midrule
    \multirow{2}{*}{5z\_vs\_1ul} &MAIC & 78.8 $\pm$ 1.2& 72.4 $\pm$ 1.9&50.2 $\pm$ 2.6& 24.0 $\pm$ 4.2& 16.6 $\pm$ 4.1\\
     &MAIC-E ($\mathcal{H}=1024$) & 70.6 $\pm$ 0.8  & 69.8 $\pm$ 1.3 & 69.5 $\pm$ 1.6 & 56.5 $\pm$ 3.5 & 34.6 $\pm$ 3.7\\
    \bottomrule
    \end{tabular}
    
    \label{tab:maic-enhance}
\end{table}

We select MAIC \cite{MAIC} as the baseline, introduce HSM with different dimensions as the transmission form on the same model, and conduct experimental comparisons.
It is shown in Tab.\ref{tab:maic-enhance} that hyperdimensional conversion significantly improves the noise robustness of MAIC in the small-scale task (hallway), and a positive correlation exists between this improvement and the message dimension $\mathcal{H}$. 
This is because HSM distributes numerical information across each dimension of the hyperdimensional space, meaning that other dimensions can mitigate significant noise interference on certain dimensions.
By contrast, for more complex tasks (5z\_vs\_1ul), the Q-table scale increases significantly, posing challenges for accurately converting the value of each digit and requiring a higher dimension.
Limited by the conversion accuracy, MAIC-E exhibits a slight performance decline under no-noise conditions, yet it exhibits stronger stability under noisy conditions.
This phenomenon can be attributed to the preservation of similarity by HSM under interference, whereby stable mapping to the corresponding digit values is maintained despite numerical alterations in dimensions.

\subsection{Extension with Policy-Based Frameworks}
This section further integrates Herm with policy-based methods to verify its scalability.
Here, MADDPG \cite{MADDPG} and MAPPO \cite{mappo} are selected as baselines, and experiments are conducted in the MPE environment under continuous-action settings.
Experimental results in Fig.\ref{fig:two-subfigs} show that Herm can effectively adapt to the two types of policy-based methods and achieve certain improvements for each. Furthermore, these results also verify that hyperdimensional encoding is equally effective for continuous observation and action spaces. Through discretization, continuous values are converted into discrete semantics, enabling a form of fuzzy control.
\begin{figure}[H]
  \centering
  \subfloat[MADDPG.]{\label{fig:suba}
    \includegraphics[width=0.34\textwidth]{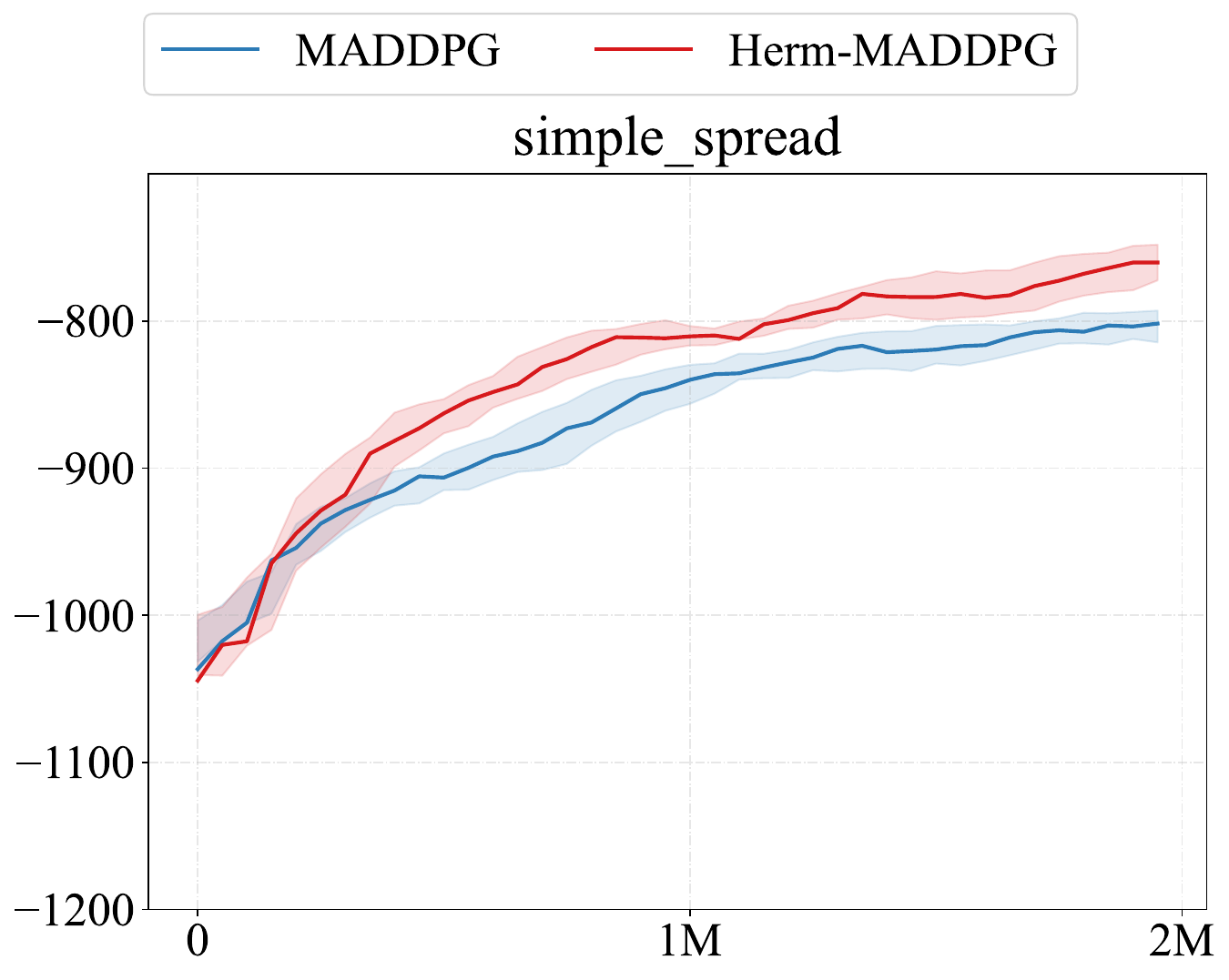}
  }
  \hfill 
  \subfloat[MAPPO.]{\label{fig:subb}
    \includegraphics[width=0.34\textwidth]{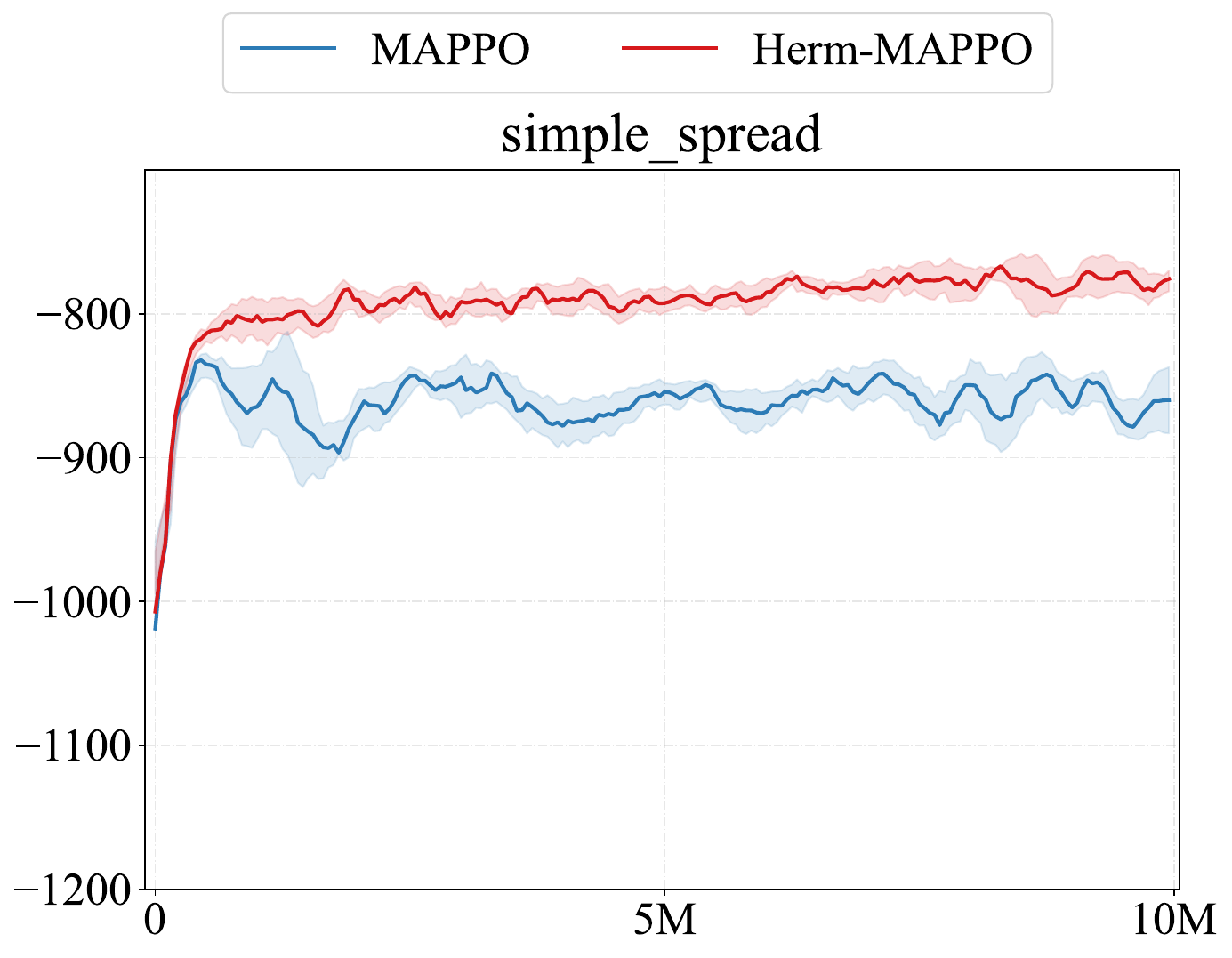}
  }
  \caption{Performance with Policy-based Extension.}
  \label{fig:two-subfigs}
\end{figure}

\subsection{Ablation Study on Message Format and Fusion Process.}

We conduct an ablation study by crossing two message formats (\texttt{Neural}, \texttt{HSM}) with two receiver fusion strategies (\texttt{Conventional}, \texttt{Herm}). All model variants adopt identical QMIX backbone, message dimension and network topology. The Neural-Herm variant (\texttt{Herm\_ne}) replaces HSM with neural encoders while preserving the dual-branch fusion mechanism. The HSM-conventional variant (\texttt{Herm\_mf}) retains HSM encoding but switches to mean fusion. We further introduce two additional ablated variants: one using neural encoding paired with mean fusion (\texttt{Herm\_ne\_mf}), and another adopting reversed similarity-based fusion (\texttt{Herm\_is}).

Experimental results show that substituting HSM with neural encoding substantially hurts model performance, leading to over $20\%$ performance degradation on the $11\times11$ task under both noise-free and noisy settings. Despite gradient-based optimization, learnable encoders directly operating on raw observations struggle to extract decision-critical information and cannot sustain stable improvement during dynamic updates. In contrast, HSM constructs a smooth manifold mapping from perception-intention to messages, effectively reducing learning stochasticity.

Moreover, both mean fusion and inverse-similarity fusion induce performance drops, where the inverse-similarity variant suffers larger degradation. Mean fusion incurs global semantic dilution, while inverse similarity yields misaligned attention: perceptually similar agents prioritize less discriminative environmental cues, whereas perceptually dissimilar agents overemphasize behavioral signals. These comparisons verify the efficacy of our dual-branch message fusion design.

We conclude that performance gains primarily originate from the HSM message format, while the fusion module acts as an auxiliary downstream component tailored for HSM.

\begin{table}[htbp]
  \centering
  \footnotesize
  \caption{Ablation of Herm's Variant under Noise}
  \label{tab:ablation_herm_noise}
  \begin{tabular}{llcccc}
    \toprule
    Task & Variant & 0.0 & 0.5 & 1.0 & 2.0 \\
    \midrule
    \multirow{5}{*}{$9\times9$}
         & Herm        & \textbf{99.1}$\pm$\textbf{0.5} & \textbf{97.6}$\pm$\textbf{1.2} & \textbf{94.8}$\pm$\textbf{0.9} & \textbf{84.7}$\pm$\textbf{1.1} \\
         & Herm\_ne    & $94.8\pm0.28$  & $89.6\pm3.9$  & $87.4\pm3.2$  & $80.5\pm4.1$ \\
         & Herm\_mf    & $93.4\pm0.4$  & $93.6\pm0.2$  & $87.8\pm1.3$  & $80.8\pm1.6$ \\
         & Herm\_is    & $92.2\pm0.7$  & $92.8\pm1.1$  & $83.6\pm0.8$  & $76.1\pm1.2$ \\
         & Herm\_ne\_mf& $90.9\pm2.5$  & $87.9\pm2.2$  & $83.7\pm3.8$  & $78.2\pm4.4$ \\
    \midrule
    \multirow{5}{*}{$11\times11$}
         & Herm        & \textbf{96.1}$\pm$\textbf{0.3} & \textbf{94.1}$\pm$\textbf{0.2} & \textbf{90.8}$\pm$\textbf{1.1} & \textbf{81.8}$\pm$\textbf{1.4} \\
         & Herm\_ne    & $74.3\pm2.4$  & $70.2\pm2.1$  & $66.8\pm3.7$  & $60.3\pm4.2$ \\
         & Herm\_mf    & $92.9\pm0.6$  & $93.1\pm0.9$  & $86.8\pm0.7$  & $79.8\pm1.5$ \\
         & Herm\_is    & $91.7\pm0.3$  & $92.3\pm0.4$  & $82.7\pm1.0$  & $76.7\pm1.3$ \\
         & Herm\_ne\_mf& $70.9\pm2.2$  & $67.9\pm3.3$  & $63.7\pm2.9$  & $58.2\pm4.0$ \\
    \bottomrule
  \end{tabular}
\end{table}

\section{Theoretical Proof for Similarity Proposition}
\label{proof}
In this section, we provide detailed elaboration and proofs for the three similarity properties involved in the decoding process. 
Specifically, the following proofs are exclusively applicable to the Holographic Reduced Representation (HRR) framework.

The original description for the three properties is as follows:
\begin{enumerate}
    \item $\delta(HV_A \oplus HV_B, HV_A) \gg 0$
    \item $\delta(HV_A \circledast HV_B, HV_A) \approx 0$
    \item $\delta(HV_A \otimes HV_B \circledast HV_B, HV_A) \gg 0$
\end{enumerate}
Since binding and unbinding are mutually inverse operations, Property 2 is equivalently expressed as: $\delta(HV_A \otimes HV_B, HV_A) \approx 0$.
Consequently, we prove Property 2 from the perspective of binding operations.

\subsection{Preliminaries and Problem Setup}

Let $\mathbb{R}^D$ denote the hyperdimensional space where $D \gg 1$. For simplicity, we denote hyperdimensional vectors by bold lowercase letters (e.g., $\mathbf{a}, \mathbf{b} \in \mathbb{R}^D$) in the following. The similarity between two vectors is measured by the cosine similarity, which simplifies to the dot product for unit-norm vectors. The Hermitian inner product is defined as $\langle \mathbf{u}, \mathbf{v} \rangle \triangleq \mathbf{u}^\top \mathbf{v}$.

Let $\mathcal{F}(\cdot)$ denote the Discrete Fourier Transform (DFT), mapping $\mathbb{R}^D \to \mathbb{C}^D$. By Parseval's theorem, the inner product is preserved (up to a scaling factor) in the frequency domain, where $\overline{\mathcal{F}(\mathbf{v})_k}$ indicates the complex conjugate of the frequency component:
\begin{equation}
    \langle \mathbf{u}, \mathbf{v} \rangle = \frac{1}{D} \sum_{k=0}^{D-1} \mathcal{F}(\mathbf{u})_k \overline{\mathcal{F}(\mathbf{v})_k}
\end{equation}

To derive the statistical bounds, we formalize the properties of atomic hypervectors.

\begin{assumption}[Hyperdimensional Random Vector]\label{ass:model}
Atomic symbols are represented by random vectors $\mathbf{v} \in \mathbb{R}^D$ with independent and identically distributed (i.i.d.) components sampled from a distribution with zero mean and variance $1/D$.
\begin{equation}
    \mathbb{E}[v_i] = 0, \quad \mathrm{Var}[v_i] = \frac{1}{D}
\end{equation}
This variance scaling ensures that $\mathbb{E}[\|\mathbf{v}\|^2] = 1$. Furthermore, distinct atomic vectors $\mathbf{u}, \mathbf{v}$ are statistically independent, implying quasi-orthogonality in high dimensions: $\mathbb{E}[\langle \mathbf{u}, \mathbf{v} \rangle] = 0$. For norm variance computability, components satisfy the fourth-moment condition $\mathbb{E}[v_i^4] = \frac{3}{D^2}$.
\end{assumption}

\begin{assumption}[Fourier Domain Statistics]\label{ass:fourier}
For large $D$, by the Central Limit Theorem (CLT, satisfied as components are i.i.d. with finite variance), the DFT components of random vectors, denoted as $z_k = \mathcal{F}(\mathbf{v})_k$, converge to a Circularly-Symmetric Complex Gaussian distribution:
\begin{equation}
    z_k \sim \mathcal{CN}(0, 1)
\end{equation}
Key statistical moments required for the following proofs are:
\begin{itemize}
    \item \textbf{Second Moment (Energy):} $\mathbb{E}[|z_k|^2] = 1$.
    \item \textbf{Fourth Moment (Kurtosis):} $\mathbb{E}[|z_k|^4] = 2$.
    \item \textbf{Independence:} For $k \neq j$, $z_k$ and $z_j$ are effectively independent (subject to conjugate symmetry constraints for real inputs, which do not affect asymptotic scaling).
\end{itemize}
\end{assumption}

\subsection{Formal Definitions of HRR Operations}

We formally define the three core algebraic operations: Bundling ($\oplus$), Binding ($\otimes$), and Unbinding ($\circledast$).

\begin{definition}[Bundling]\label{def:superposition}
Bundling aggregates a set of vectors $\{\mathbf{x}_i\}_{i=1}^n$ into a single composite vector $\mathbf{S}$. To prevent norm growth, we define the normalized bundling as Equ.\ref{equ:bundle}. This operation preserves similarity. For any component $\mathbf{x}_k$, the expected similarity $\mathbb{E}[\langle \mathbf{S}, \mathbf{x}_k \rangle] > 0$.
\begin{equation}
    \mathbf{S} \triangleq \frac{1}{\sqrt{n}} \sum_{i=1}^n \mathbf{x}_i.
\label{equ:bundle}
\end{equation}
 
\end{definition}

\begin{definition}[Binding via Circular Convolution]\label{def:binding}
Binding associates two vectors $\mathbf{a}$ and $\mathbf{b}$ into a pair $\mathbf{c}$. In HRR, this is realized by circular convolution, denoted as $\mathbf{c} = \mathbf{a} \otimes \mathbf{b}$. 
There exist two mathematical formulations for binding in the spatial domain and frequency domain as follows, where $\circ$ denotes the element-wise (Hadamard) product.
Binding creates a vector that is quasi-orthogonal to its constituents ($\mathbf{c} \perp \mathbf{a}$ and $\mathbf{c} \perp \mathbf{b}$).
\begin{itemize}
    \item \textbf{Spatial Domain:} $c_k = \sum_{j=0}^{D-1} a_j b_{(k-j) \pmod D}$.
    \item \textbf{Frequency Domain:} $\mathcal{F}(\mathbf{c}) = \mathcal{F}(\mathbf{a}) \circ \mathcal{F}(\mathbf{b})$.
\end{itemize}

\end{definition}

\begin{definition}[Unbinding via Approximate Inverse]\label{def:unbinding}
Unbinding refers to the operation that extracts a constituent vector from a bound pair. The exact inverse of circular convolution suffers from numerical instability. HRR adopts involution as the approximate inverse, which corresponds to the complex conjugate in the frequency domain.
Let $\mathbf{b}^\dagger$ denote the involution of $\mathbf{b}$, defined in the frequency domain as $\mathcal{F}(\mathbf{b}^\dagger)_k = \overline{\mathcal{F}(\mathbf{b})_k}$.
The unbinding operation is defined as circular correlation (denoted $\circledast$):
\begin{equation}
    \mathbf{a}' \triangleq \mathbf{c} \otimes \mathbf{b}^\dagger \triangleq \mathbf{c} \circledast \mathbf{b}
\end{equation}
In the frequency domain:
\begin{equation}
    \mathcal{F}(\mathbf{a}')_k = \mathcal{F}(\mathbf{c})_k \cdot \overline{\mathcal{F}(\mathbf{b})_k}
\end{equation}
The unbinding operation is inherently noisy: the circular convolution of $\mathbf{b}$ and its involution $\mathbf{b}^\dagger $ yields $\mathbf{b} \otimes \mathbf{b}^\dagger$, which merely approximates the identity impulse $\boldsymbol{\delta}$; such approximation inevitably induces a residual noise term, and the magnitude of this term is strictly determined by the deviation of $|\mathcal{F}(\mathbf{b})_k|^2$ from unity.
\end{definition}

\subsection{Similarity Properties}

\subsubsection{Bundling}
Bundling aggregates information while preserving similarity to constituents.

\begin{theorem}[Bundling Statistics]
For $D \to \infty$ and $n \ll D$, let $\mathbf{S} = \frac{1}{\sqrt{n}} \sum_{i=1}^n \mathbf{a}_i$. For any $k \in \{1,\dots,n\}$:
\begin{equation}
    \mathbb{E}[\mathbf{S} \cdot \mathbf{a}_k] = \frac{1}{\sqrt{n}}, \quad \text{Var}[\mathbf{S} \cdot \mathbf{a}_k] \approx \frac{1}{D}
\end{equation}
\end{theorem}

\begin{proof}

Let $\mathbf{S} = \frac{1}{\sqrt{n}}\sum_{i=1}^n\mathbf{a}_i$, decompose the dot product and calculate the expectation:
\begin{equation}
\begin{split}
\mathbf{S} \cdot \mathbf{a}_k &= \left( \frac{1}{\sqrt{n}}\sum_{i=1}^n\mathbf{a}_i \right) \cdot \mathbf{a}_k = \frac{1}{\sqrt{n}} \left( \|\mathbf{a}_k\|^2 + \sum_{i \neq k} \mathbf{a}_i \cdot \mathbf{a}_k \right)\\
\end{split}
\end{equation}
By Assumption \ref{ass:model}, $\mathbb{E}[\|\mathbf{a}_k\|^2] = 1$ and the cross terms have zero mean, i.e., $\mathbb{E}[\mathbf{a}_i \cdot \mathbf{a}_k] = 0 \ (i \neq k)$. 
\begin{equation}
\begin{split}
\mathbb{E}[\mathbf{S} \cdot \mathbf{a}_k] &= \frac{1}{\sqrt{n}} \left( \mathbb{E}[\|\mathbf{a}_k\|^2] + \sum_{i \neq k} \mathbb{E}[\mathbf{a}_i \cdot \mathbf{a}_k] \right) 
= \frac{1}{\sqrt{n}} \left( 1 + \sum_{i \neq k} 0 \right) = \frac{1}{\sqrt{n}}. \\
\end{split}
\end{equation}

Let $\mathbf{a}_k = [a_{k1},a_{k2},\dots,a_{kD}]^T$, the step-by-step calculation is as follows:

Expand by the definition of norm and combine with the variance property $\text{Var}(X)=\mathbb{E}[X^2]-\mathbb{E}^2[X]$:
\begin{equation}
\|\mathbf{a}_k\|^2 = \sum_{i=1}^D a_{ki}^2, 
\end{equation}
\begin{equation}
    \text{Var}(\|\mathbf{a}_k\|^2) = \text{Var}\left( \sum_{i=1}^D a_{ki}^2 \right) = \sum_{i=1}^D \left( \mathbb{E}[a_{ki}^4] - \mathbb{E}^2[a_{ki}^2] \right) 
= \sum_{i=1}^D \left( \frac{3}{D^2} - \left( \frac{1}{D} \right)^2 \right) = \frac{2}{D}.
\end{equation}

Since $\mathbf{a}_i$ and $\mathbf{a}_k$ are mutually independent, thus:
$\text{Var}(\mathbf{a}_i \cdot \mathbf{a}_k) = \frac{1}{D}$.

\begin{equation}
\begin{split}
\text{Var}[\mathbf{S} \cdot \mathbf{a}_k] &= \text{Var}\left[ \frac{1}{\sqrt{n}} \left( \|\mathbf{a}_k\|^2 + \sum_{i \neq k} \mathbf{a}_i \cdot \mathbf{a}_k \right) \right] 
= \frac{1}{n} \text{Var}(\|\mathbf{a}_k\|^2) + \frac{1}{n} \text{Var}\left( \sum_{i \neq k} \mathbf{a}_i \cdot \mathbf{a}_k \right) \\
&= \frac{\text{Var}(\|\mathbf{a}_k\|^2)}{n} + \frac{n-1}{n} \text{Var}(\mathbf{a}_i \cdot \mathbf{a}_k) 
= \frac{2}{nD} + \frac{n-1}{nD} = \frac{n+1}{nD} \\
\end{split}
\end{equation}

When $n \ll D$, the term $\frac{2}{nD}$ is negligible, thus:
$\text{Var}[\mathbf{S} \cdot \mathbf{a}_k] \approx \frac{1}{D}.$
\end{proof}

\subsubsection{Binding}
Binding creates a composite representation orthogonal to its constituents.

\begin{theorem}[Binding Orthogonality]
For $D \to \infty$, let $\mathbf{c} = \mathbf{a} \otimes \mathbf{b}$. Then $\mathbb{E}[\mathbf{c} \cdot \mathbf{a}] = 0$ and $\text{Var}[\mathbf{c} \cdot \mathbf{a}] = 2/D$.
\end{theorem}

\begin{proof}
In the Fourier domain, $\mathcal{F}(\mathbf{c})_k = \mathcal{F}(\mathbf{a})_k \mathcal{F}(\mathbf{b})_k$. By Parseval's theorem:
\begin{equation}
    \mathbf{c} \cdot \mathbf{a} = \frac{1}{D} \sum_{k=0}^{D-1} |\mathcal{F}(\mathbf{a})_k|^2 \mathcal{F}(\mathbf{b})_k
\end{equation}
Since $\mathbf{a},\mathbf{b}$ are real-valued, $\mathcal{F}(\mathbf{v})_{D-k} = \overline{\mathcal{F}(\mathbf{v})_k}$, ensuring the sum is real. By Assumption \ref{ass:fourier}, $\mathbb{E}[\mathcal{F}(\mathbf{b})_k]=0$ and vectors are independent, so $\mathbb{E}[\mathbf{c} \cdot \mathbf{a}] = 0$.

For variance, let $X_k = |\mathcal{F}(\mathbf{a})_k|^2 \mathcal{F}(\mathbf{b})_k$ (real-valued due to conjugate symmetry). $\text{Var}[X_k] = \mathbb{E}[|\mathcal{F}(\mathbf{a})_k|^4]\mathbb{E}[|\mathcal{F}(\mathbf{b})_k|^2] = 2 \cdot 1 = 2$. Summing over $D$ independent components:
\begin{equation}
    \text{Var}[\mathbf{c} \cdot \mathbf{a}] = \frac{1}{D^2} \sum_{k=0}^{D-1}  \text{Var}[X_k] = \frac{1}{D^2} \cdot 2D = \frac{2}{D}.
\end{equation}
\end{proof}

\subsubsection{Unbinding}
Unbinding retrieves the original vector from the bound pair using the approximate inverse.

\begin{theorem}[Unbinding Recovery and SNR]\label{thm:unbind}
For $D \to \infty$, let $\mathbf{a}' = (\mathbf{a} \otimes \mathbf{b}) \circledast \mathbf{b}$. The retrieval similarity follows:
\begin{equation}
    \mathbb{E}[\mathbf{a}' \cdot \mathbf{a}] = 1, \quad Var[\mathbf{a}' \cdot \mathbf{a}] = \frac{6}{D}.
\end{equation}
Consequently, the Signal-to-Noise Ratio (SNR) scales as $\mathcal{O}(D)$.
\end{theorem}

\begin{proof}
Since the unbinding operation is diagonalized in the Fourier domain, we analyze the statistics of the frequency components.
While real-valued vectors impose Hermitian symmetry ($\mathcal{F}(\mathbf{a})_k = \overline{\mathcal{F}(\mathbf{a})_{D-k}}$), we treat the components as independent for asymptotic analysis. This simplifies the derivation without affecting the order of magnitude of the SNR.
The unbinding operation in the Fourier domain is $\mathcal{F}(\mathbf{a}')_k = \mathcal{F}(\mathbf{a})_k |\mathcal{F}(\mathbf{b})_k|^2$. Let $Z_k = |\mathcal{F}(\mathbf{a})_k|^2 |\mathcal{F}(\mathbf{b})_k|^2$. The similarity is the sample mean of $Z_k$:
\begin{equation}
    \mathbf{a}' \cdot \mathbf{a} = \frac{1}{D} \sum_{k=0}^{D-1} Z_k.
\end{equation}

\noindent \textbf{Expectation:}
Using the independence of $\mathbf{a}$ and $\mathbf{b}$:
\begin{equation}
    \mathbb{E}[Z_k] = \mathbb{E}[|\mathcal{F}(\mathbf{a})_k|^2] \mathbb{E}[|\mathcal{F}(\mathbf{b})_k|^2] = 1 \cdot 1 = 1 \implies \mathbb{E}[\mathbf{a}' \cdot \mathbf{a}] = 1
\end{equation}

\noindent \textbf{Variance:}
First, we compute the variance of a single frequency component $Z_k = |\mathcal{F}(\mathbf{a})_k|^2 |\mathcal{F}(\mathbf{b})_k|^2$. Since the elements of $\mathbf{a}, \mathbf{b}$ are i.i.d. Gaussian, the frequency components converge to a complex Gaussian distribution (CLT). Using the fourth-moment property $\mathbb{E}[|z|^4]=2$ for standard complex Gaussian variables:
\begin{equation}
    \text{Var}[Z_k] = \mathbb{E}[Z_k^2] - (\mathbb{E}[Z_k])^2 = \mathbb{E}[|\mathcal{F}(\mathbf{a})_k|^4] \mathbb{E}[|\mathcal{F}(\mathbf{b})_k|^4] - 1^2 = (2 \cdot 2) - 1 = 3. 
\end{equation}
For the frequency components, the variance is formulated as:
\begin{equation}
        \text{Var}[\mathbf{a}' \cdot \mathbf{a}] = \text{Var}\left[\frac{1}{D}\sum_{k=0}^{D-1} Z_k\right] = \frac{1}{D^2} \cdot 3D = \frac{3}{D}. 
\end{equation}

For real-valued vectors, the Fourier spectrum satisfies Hermitian symmetry: $\mathcal{F}(\mathbf{x})_k = \overline{\mathcal{F}(\mathbf{x})}_{D-k}$, which implies $Z_k = Z_{D-k}$. Consequently, the components in the sum are pairwise correlated. The variance of the mean is:
\begin{equation}
    \text{Var}[\mathbf{a}' \cdot \mathbf{a}] = \text{Var}\left[\frac{1}{D}\sum_{k=0}^{D-1} Z_k\right] = \frac{1}{D^2} \sum_{k=0}^{D-1} \sum_{j=0}^{D-1} \text{Cov}[Z_k, Z_j]
\end{equation}
Due to the independence of distinct frequencies, $\text{Cov}[Z_k, Z_j] = 0$ for $j \notin \{k, D-k\}$. The non-zero covariance terms are the variance itself ($j=k$) and the symmetry term ($j=D-k$). Since $Z_k = Z_{D-k}$, it results in:
\begin{equation}
    \text{Cov}[Z_k, Z_{D-k}] = \text{Var}[Z_k] = 3
\end{equation}
Thus, for each $k$, the inner sum contributes $2 \times \text{Var}[Z_k]$:
\begin{equation}
    \text{Var}[\mathbf{a}' \cdot \mathbf{a}] \approx \frac{1}{D^2} \sum_{k=0}^{D-1} (\text{Var}[Z_k] + \text{Cov}[Z_k, Z_{D-k}]) = \frac{1}{D^2} \sum_{k=0}^{D-1} (3 + 3) = \frac{6}{D}
\end{equation}

\noindent \textbf{SNR:}
The SNR is defined as the ratio of the squared signal expectation to the noise variance:
\begin{equation}
    \text{SNR} = \frac{(\mathbb{E}[\mathbf{a}' \cdot \mathbf{a}])^2}{\text{Var}[\mathbf{a}' \cdot \mathbf{a}]} = \frac{1}{6/D} = \frac{D}{6}
\end{equation}
\end{proof}

\begin{remark}
A critical feature of HRR is the $\mathcal{O}(D)$ scaling of SNR: as the dimension $D$ increases, SNR grows linearly.
It implies the multiplicative noise introduced by the approximate inverse is \textit{effectively suppressed in high-dimensional spaces}. This enables reliable symbol retrieval without exact inverse computation.

By Chebyshev's inequality, for any $\epsilon > 0$,
$$\mathbb{P}\left( \vert \mathbf{a}' \cdot \mathbf{a} - 1 \vert \geq \epsilon \right) \leq \frac{\text{Var}[\mathbf{a}' \cdot \mathbf{a}]}{\epsilon^2} = \frac{6}{D\epsilon^2}.$$
As $D \to \infty$, the right tends to 0, so $\mathbf{a}' \cdot \mathbf{a} \xrightarrow{p} 1$ (converges in probability to 1), confirming the \textit{asymptotic retrieval accuracy} in high dimensions.
\end{remark}

 \clearpage
\section{Conversion Fidelity Analysis on Message Length}
\label{Interpretability-Analysis}
In this section, we investigate the impact of message length on the accuracy of semantic transmission. 
We dynamically adjust the message length $\mathcal{H}$ in Herm and perform message decoding for identical scenarios, as visualized in Fig.\ref{fig:smac_explain} (SMAC) and Fig.\ref{fig:explain_h} (LBF).
The red boxes indicate original semantics, while the black boxes denote erroneously decoded semantics deviating from the originals. The results show that when $\mathcal{H}$ is small, the decoded outputs contain significant errors. Furthermore, similarity metrics exhibit random patterns without discernible trends, contradicting the value dictionary's design principle. We attribute these phenomena to severe degradation of vector orthogonality, which violates the pseudo-inverse requirement for unbinding in hyperdimensional computing. As $\mathcal{H}$ increases, the error rate drops markedly, with a near-perfect decoding achieved at $\mathcal{H}$=512.
Meanwhile, similarity scores within the same value dictionary exhibit a distinct peak phenomenon, validating the rationality of the encoding-decoding process.

\begin{figure*}[!htb]
     \centering
     \includegraphics[width=0.95\textwidth]{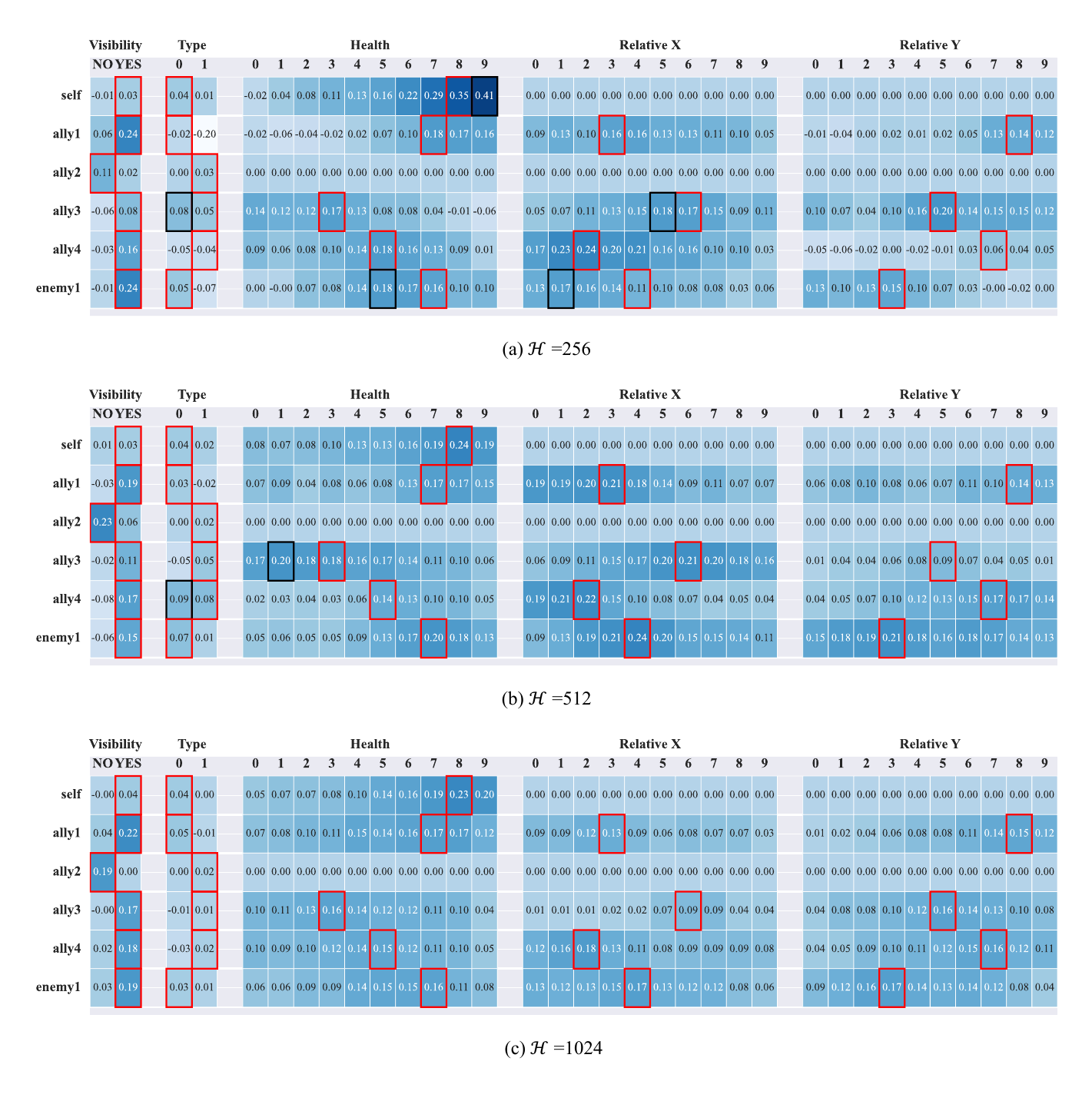}
     \caption{Conversion Fidelity in SMAC with Different Message Length.}
     \label{fig:smac_explain}
\end{figure*}
\clearpage
\begin{figure*}[!htb]
     \centering
     \includegraphics[width=0.9\textwidth]{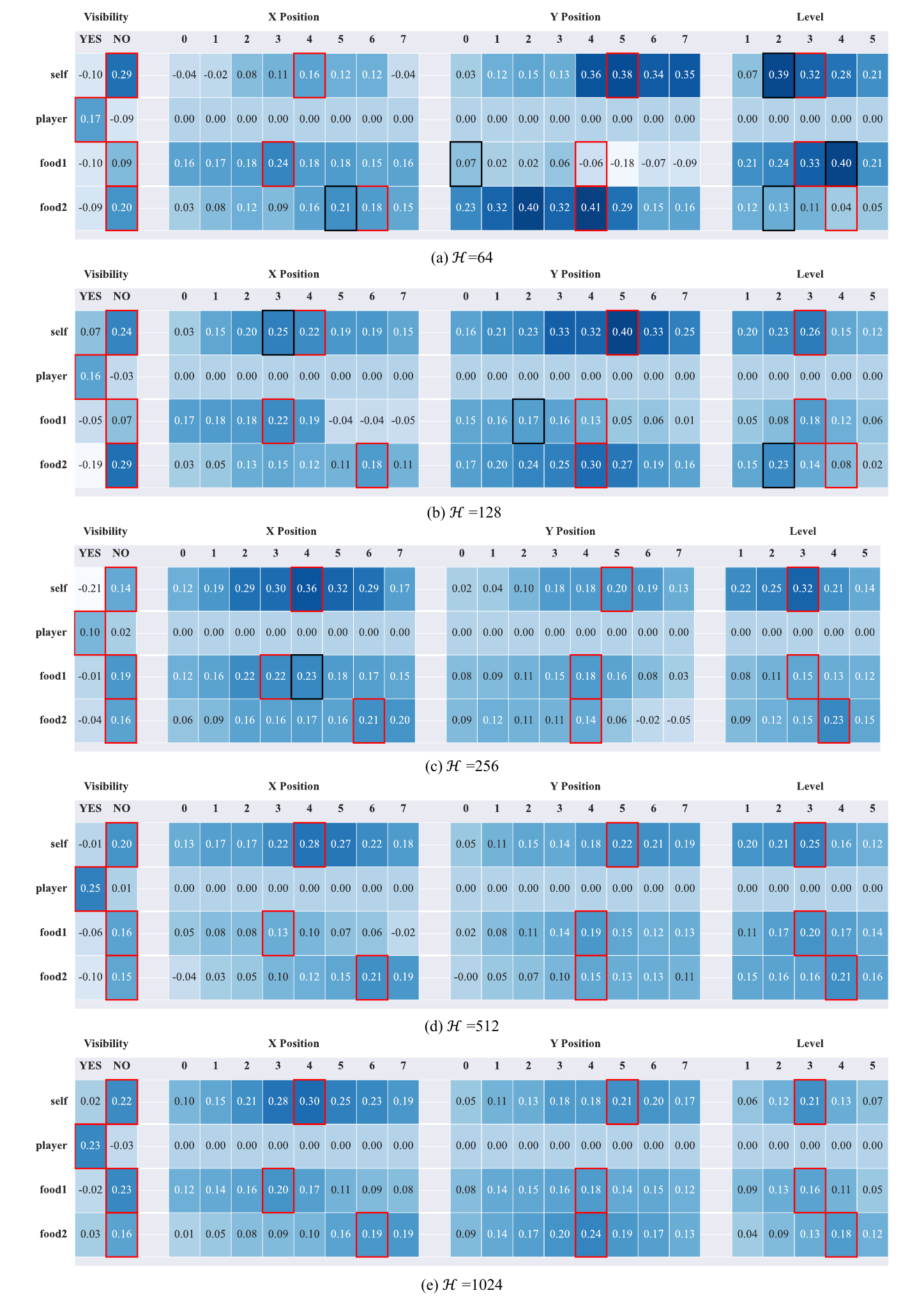}
     \caption{Conversion Fidelity in LBF with Different Message Length.}
     \label{fig:explain_h}
\end{figure*}

\begin{figure}[h]
    \centering
    \includegraphics[width=0.9\linewidth]{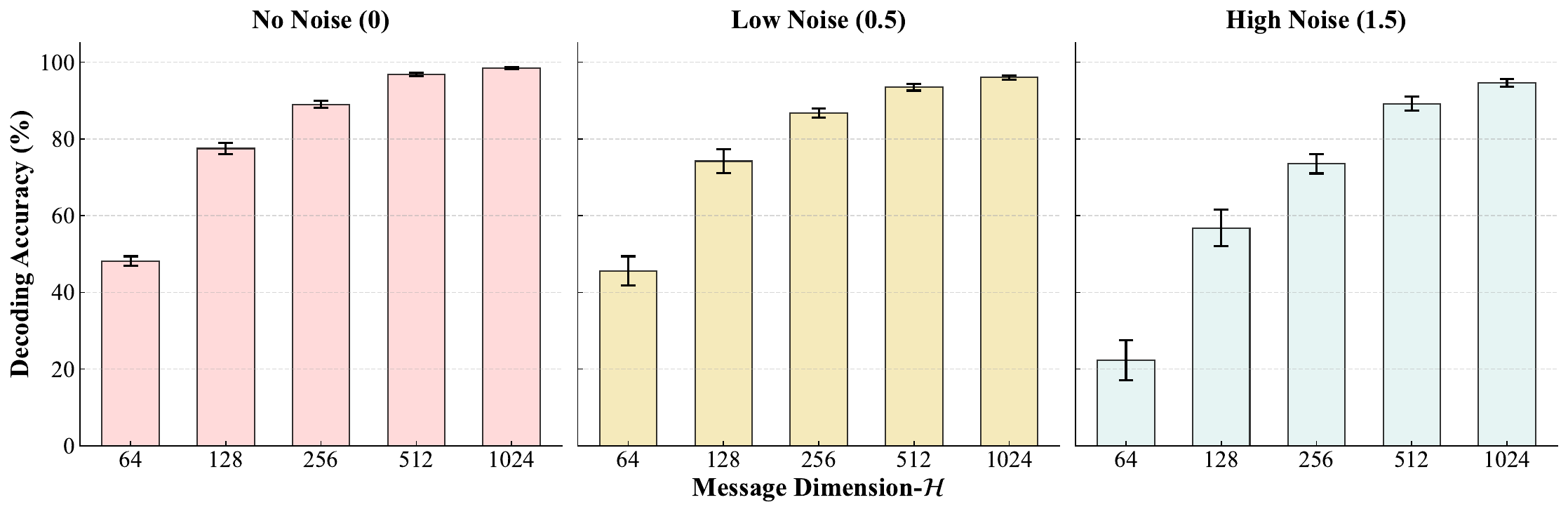}
    \caption{Decoding Accuracy in SMAC with Different Message Length.}
    \label{fig:dec-acc}
\end{figure}
Fig.\ref{fig:dec-acc} investigates the semantic discrepancy between decoded hyperdimensional symbolic messages (HSM) and original semantics under noisy settings ($\eta_r$).
Here, we perform iterative decoding on one hundred messages and compare the transformed semantics with the original semantics.
Results show that message dimension $\mathcal{H}$ regulates the semantic discrepancy: increased decoding accuracy and decreased variance occur with the rise of message dimension.
This indicates that the conversion fidelity of HSM is significantly enhanced with the increase in dimension. Furthermore, message dimension also significantly affects decoding robustness—when the dimension exceeds 256, both the mean accuracy and variance remain stable.

\section{Message Similarity Analysis under Noise}
\label{Similarity-Analysis}
We conduct a similarity evaluation on the observation component of communication messages in SMAC, with results illustrated in Fig.\ref{fig:smac_attention} (b). The findings indicate that the observation encoding of the rightmost stalker exhibits the highest similarity to that of the middle stalker (excluding itself). This aligns with their actual scenario characteristics in Fig.\ref{fig:smac_attention} (a): they belong to the same race and possess similar relative positions with respect to other entities. This phenomenon reveals that hyperdimensional vectors can still effectively preserve semantic similarity even after undergoing multiple bundling and binding. Furthermore, under different noise conditions, the correspondence between semantic similarity and encoding similarity remains intact, as shown in Fig.\ref{fig:smac_attention} (c-d). This further illustrates the robustness of the proposed hyperdimensional message encoding in this paper.
\begin{figure*}[!htb]
     \centering
     \includegraphics[width=0.9\textwidth]{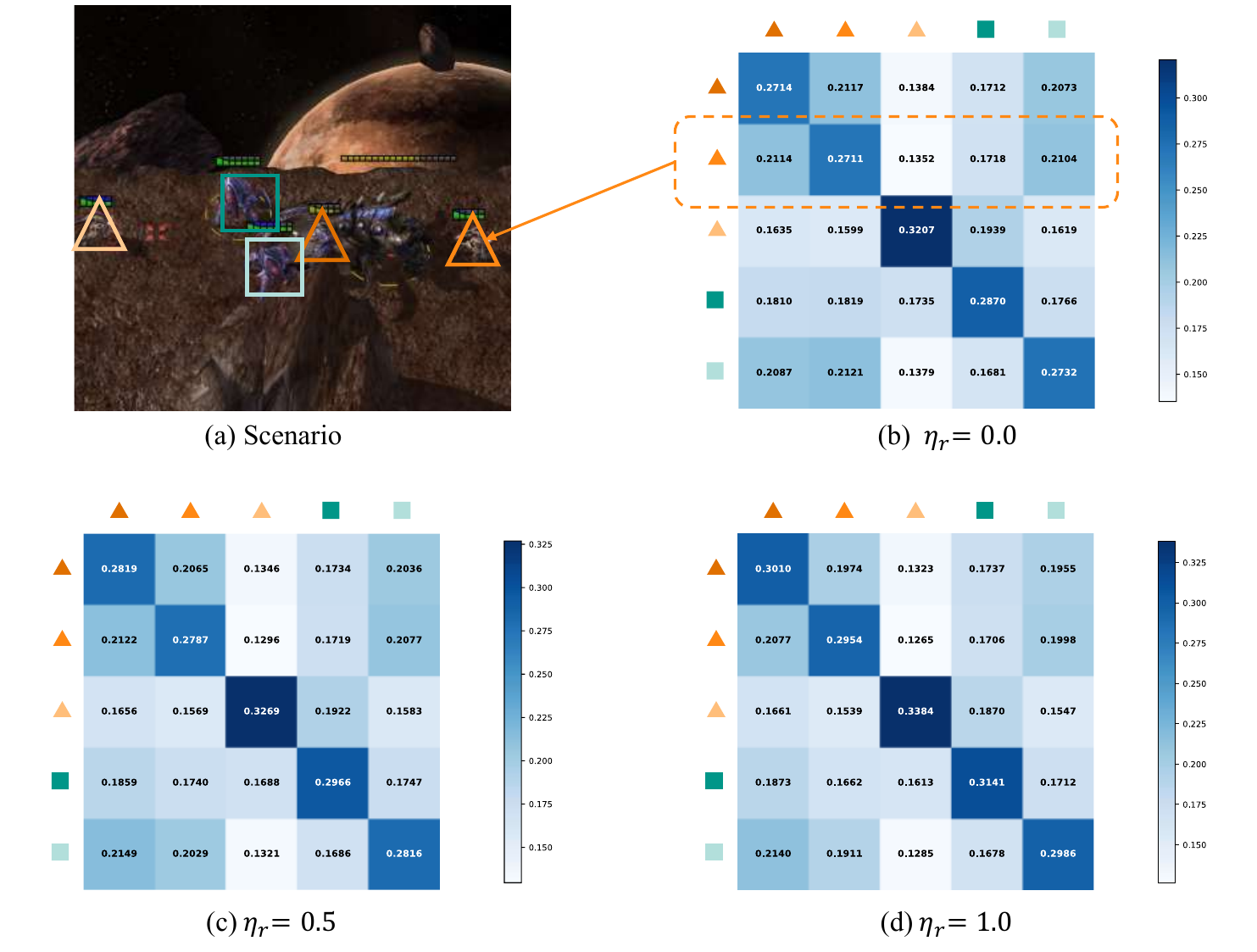}
     \caption{Observation Similarity Visualization in SMAC. Triangles denote ``stalker'' and squares denote ``zealot''. Subfigures b, c, and d present the observation encoding similarity matrices under different noise conditions.}
     \label{fig:smac_attention}
\end{figure*}

\section{Similarity-Preservation Verification for Ego-Binding}
\label{Similarity-Preservation}
Fig. \ref{fig:similarity_matrix} compares the cosine similarity matrices of agent messages before and after ego-binding. The two matrices show nearly identical global patterns and value distributions, with perfect preservation of mutual similarity and no semantic inversion observed.
While ego-binding leads to a rise in the absolute similarity values, the relative similarity relations across agents are preserved. This indicates that semantic alignment is consistent with weight calculation and incurs no conflict.
 \begin{figure*}[!htb]
     \centering
     \includegraphics[width=0.9\textwidth]{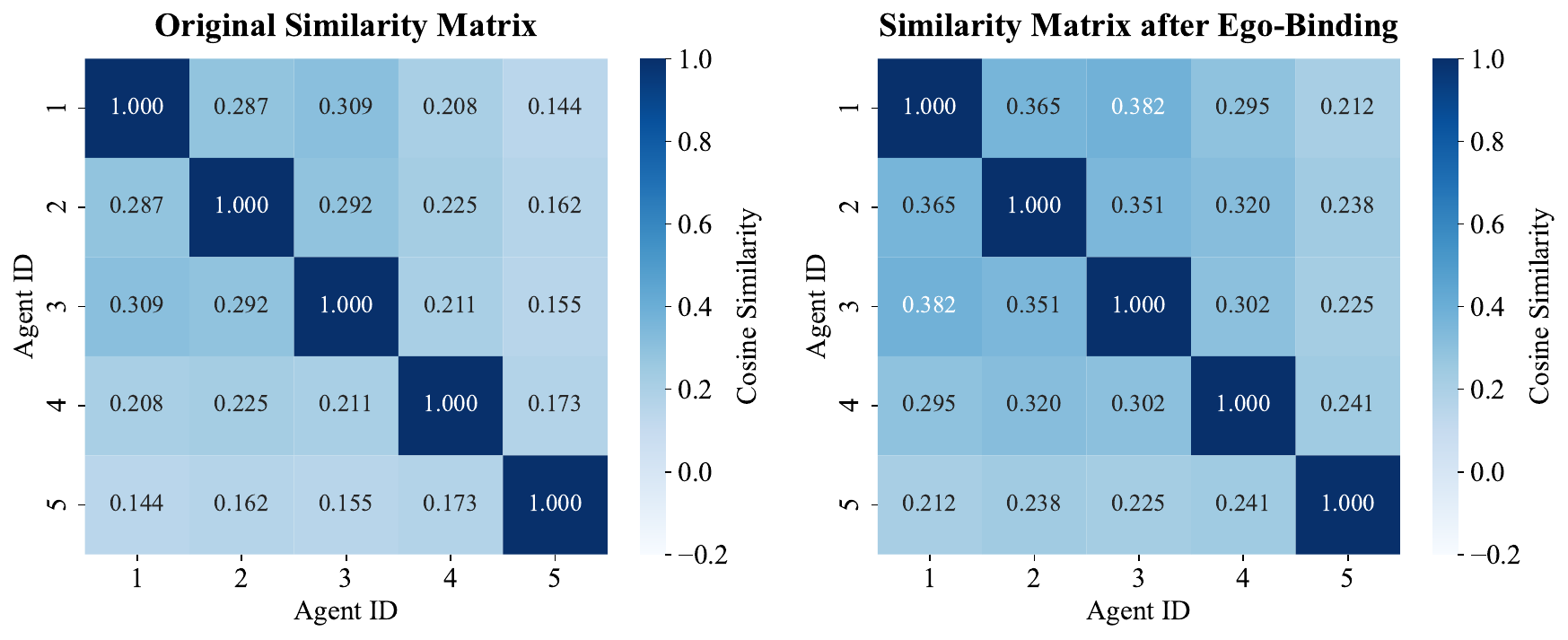}
     \caption{Analysis of the Ego-Binding on Similarity. }
     \label{fig:similarity_matrix}
\end{figure*}

\section{Hyperdimensional Encoding Process}
In this section, we elaborate on the detailed procedures of hyperdimensional message encoding in each experimental environment, including attribute partitioning, value constraints, and computational operations.
\subsection{SMAC / SMACv2}
Despite variations in environmental layouts and observational capacities, SMAC and SMACv2 maintain identical fundamental observation and action structures at their core.
\subsubsection{Attribute Dictionary.}
We posit mutual independence among attributes and construct the encoding dictionary through stochastic sampling.
The observation is formally defined as an attribute dictionary. Tab \ref{tab:attributes} introduces the detailed design. Notably, the attribute design here is merely intended for human readability. For agents, it only performs dimension-wise encoding over raw observations.
\begin{table}[htbp]
\centering
\footnotesize
\caption{Attribute Properties Specification in SMAC}
\label{tab:attributes}
\begin{tabular}{lcccc}
\toprule
\textbf{Attribute} & \textbf{Explanation} & \textbf{Independence} & \textbf{Discrete} & \textbf{Value Range} \\
\midrule
Agent ID &  Unique Identifier  & Independent    & Yes               & $\{1,\cdots,N_{entity}\}$ \\
HP &  Current Health Value    & Associated      & No                & $[0, 1]$ \\
Relative X & Normalized X-position Relative to Agent & Associated  & No            & $(-1, 1)$ \\
Relative Y & Normalized Y-position Relative to Agent & Associated  & No            & $(-1, 1)$ \\
Unit Type & Race of the Unit  & Independent & Yes             & $\{1,\cdots,N_{type}\}$ \\
Visibility & Binary Visibility Indicator  & Independent           & Yes            & $\{0,1\}$ \\
Identity & Indicator Whether the Unit is Self, Ally, or Enemy  & Independent    & Yes & $\{0,1,2\}$ \\
Time & Time Window Index & Associated & No  & $\{-1,\cdots,-w\}$ \\
\bottomrule
\end{tabular}
\end{table}

The action encoding comprises all executable actions: NOTHING, STAY, MOVE\_UP, MOVE\_DOWN, MOVE\_LEFT, MOVE\_RIGHT, ATTACK, HEAL.
\subsubsection{Value Dictionary.}
Different attributes correspond to distinct value types, such as independent/associated and discrete/continuous. We provide detailed specifications for each attribute in Table \ref{tab:attributes}. For independent values, we employ random sampling. For continuous variables, we perform interval discretization followed by cyclic replacement encoding, typically using 10 bins.

\subsubsection{Encoding Process.} The observation encoding follows an entity-wise paradigm, where self features, ally features, and enemy features are separately encoded before being bundled into a composite hyperdimensional vector. Here, $HV_{attr}$ denotes the result after binding attribute-value pairs.
\begin{itemize}
    \item For self-encoding, we bind the agent's ID, Identity, HP, and Type. The process is formulated as $HV_{id} \otimes HV_{self} \otimes (HV_{hp} \oplus HV_{type})$. A zero vector is returned if the current agent is terminated.
    \item For ally/enemy encoding, we first perform visibility verification before feature extraction.
    If invisible, the process is defined as $HV_{id} \otimes HV_{ally/enemy} \otimes HV_{invisible}$. If visible, we bundle its static attributes and perform identity binding, formulated as $HV_{id} \otimes HV_{ally/enemy} \otimes (HV_{visible} \oplus HV_{hp} \oplus HV_{x} \oplus HV_{y} \oplus HV_{type})$.
\end{itemize}
For action sequence encoding, we perform temporal binding of actions within the window and then bundle all features, which is denoted as $(HV_{-1} \otimes HV_{a_{-1}}) \oplus \cdots \oplus (HV_{-w} \otimes HV_{a_{-w}})$

\subsection{Lbforaging (LBF)}
Due to the similarity in the encoding process, we only present the attribute dictionary and value dictionary (Table \ref{tab:attributes:LBF}) here, while the encoding procedure follows the aforementioned section.

\begin{table}[htbp]
\centering
\footnotesize
\caption{Attribute Properties Specification in LBF}
\label{tab:attributes:LBF}
\begin{tabular}{lcccc}
\toprule
\textbf{Attribute} & \textbf{Explanation} & \textbf{Independence} & \textbf{Discrete} & \textbf{Value Range} \\
\midrule
Agent ID  &   Unique Identifier        & Independent          & Yes               & $\{1,\cdots,N_{entity}\}$ \\
Absolute X  & Absolute X-position in the Map     & Associated           & Yes                & $\{0,\cdots,X\}$ \\
Absolute Y & Absolute Y-position in the Map       & Associated           & Yes                & $\{0,\cdots,Y\}$ \\
Level  &  Capture Ability (Predator) / Escape Ability (Prey)       & Independent          & Yes               & $\{1,\cdots,Level\}$ \\
Visibility & Binary Visibility Indicator        & Independent           & Yes               & $\{0,1\}$ \\
Identity & Indicator Whether the Unit is Self, Predator, or Prey      & Independent          & Yes                & $\{0,1,2\}$ \\
Time & Time Window Index & Associated           & No  & $\{-1,\cdots,-w\}$ \\
\bottomrule
\end{tabular}
\end{table}

The action encoding comprises all executable actions: NOTHING, STAY, MOVE\_UP, MOVE\_DOWN, MOVE\_LEFT, MOVE\_RIGHT, CAPTURE.

\subsection{Hallway}
Due to the similarity in the encoding process, we only present the attribute dictionary and value dictionary (Table \ref{tab:attributes:hallway}) here, while the encoding procedure follows the aforementioned section.

\begin{table}[htbp]
\centering
\footnotesize
\caption{Attribute Properties Specification in Hallway}
\label{tab:attributes:hallway}
\begin{tabular}{lcccc}
\toprule
\textbf{Attribute} & \textbf{Explanation} & \textbf{Independence} & \textbf{Discrete} & \textbf{Value Range} \\
\midrule
Agent ID &  Unique Identifier         & Independent          & Yes               & $\{1,\cdots,N_{agent}\}$ \\
Position & Position in its Line      & Associated           & Yes                & $\{0,\cdots,X\}$ \\
Time & Time Window Index & Associated           & No  & $\{-1,\cdots,-w\}$ \\
\bottomrule
\end{tabular}
\end{table}

The action encoding comprises all executable actions: STAY, MOVE\_LEFT, MOVE\_RIGHT. 

\section{Collaboration with Large Language Models}
\label{llm+comm}
Hyperdimensional Symbolic Message (HSM) introduced in our paper constructs a bidirectional conversion pipeline that links (1) message encoding, (2) symbolic semantics, and (3) natural language. 
In this framework, message encoding enables efficient agent-to-agent communication and decision-making processes. Natural language supports interpretability for both human operators and large language models (LLMs). Symbolic semantics then acts as a core bridging medium that connects these two functionally distinct components.
Such explicit conversion fidelity eliminates the algorithm’s reliance on foundation model embedding spaces, while ensuring compatibility with various LLMs.
Next, we detail the communication and collaborative decision-making processes between agents and LLMs.

\subsection{Message Transformation}
\begin{figure}[ht]
    \centering
    \begin{lstlisting}[style=promptstyle]
I'm Agent 2, my type is 1, and my HP is within 0.40-0.50. 
I can not see Ally_0.
I can see Ally_1, its type is 0, its HP is within 0.40~0.50, its relative position x to me is within -0.80~-0.60, and its relative position y to me is within -0.80~-0.60.
I can see Ally_3, its type is 0, its HP is within 0.40~0.50, its relative position x to me is within -1.00~-0.80, and its relative position y to me is within -0.60~-0.40. 
I can see Ally_4, its type is 1, its HP is within 0.40~0.50, its relative position x to me is within -0.60~-0.40, and its relative position y to me is within -1.00~-0.80.
I can not see Enemy_0.
My previous actions contain ['Move West', 'Move North'].
    \end{lstlisting}
    \caption{Example of Decoded Message in SMAC}
    \label{fig:llm-message}
    \end{figure}
For agents, the messages during communication are actually hyperdimensional vectors.
Through the proposed iterative decoding method (Section \ref{section-decode}), these hyperdimensional vectors are decoded into a series of structured semantic entities, as depicted in Fig.\ref{fig:smac_explain} and Fig.\ref{fig:explain_h}. 
The semantic entities are expressed as attribute-value pairings, which can then be aggregated and mapped to natural language.
We present the form of messages following decoding and programmatic transformation in Fig.\ref{fig:llm-message}, which contains entity-wise information.

\subsection{Decision Making Process}
In the team, a single agent is designated as an L-Agent, which is controlled by LLMs; all remaining agents are defined as A-Agents that operate under policies. The joint action of both types of agents is transmitted and executed in the environment. 

\textbf{A-Agent $\rightarrow$ A-Agent}.
During each decision time step, every A-Agent encodes its local perceptual data and historical action sequences into hyperdimensional symbolic messages (HSMs) for inter-agent communication and coordination. 

\textbf{A-Agent $\rightarrow$ L-Agent}.
Messages from A-Agents are uniformly translated into natural language, aggregated into a structured report, and forwarded to the L-Agent for further processing.
The L-Agent then selects the current decision from the set of executable actions, guided by the reasoning process of LLMs.

\textbf{L-Agent $\rightarrow$ A-Agent}.
The L-Agent is also required to convey its local observations and historical action sequences to peer agents.
 Given that A-Agents are designed only to accept HSMs, the L-Agent utilizes the same hyperdimensional encoding module as A-Agents.
 The encoded messages are subsequently delivered to A-Agents to facilitate their decision-making processes.

We design the prompt to guide LLM in agent collaboration as Fig.\ref{fig:llm-prompt}.
\begin{figure}[h]
    \centering
    \begin{lstlisting}[style=promptstyle]
Core Role & Positioning:
We are playing StarCraft II micro scenario "2s3z_vs_1ul", trying to control our agents to defeat all of the enemy units. 
You are the last member of the multi-agent team, numbered Agent-4, collaborating with 4 algorithm-controlled agents (numbered Agent-0, Agent-1, ..., Agent-3) to complete tasks.
You must make decisions strictly based on messages sent by other agents, with no access to additional global information.

Input Format:
In each step, you will perceive your own observation with natural language and receive structured messages from all agents in the team.
Each message contains the following fields: [Agent-ID] | [Current Observation] | [Previous Actions Sequence].
Furthermore, you will receive your current available action list. Action index is listed as: {
0: No-op only available for dead agents; 1: Stop; 2: Move North; 3: Move South; 4: Move East; 5: Move West; 6 to 6+n_enemy: Attack the enemy with a certain index (index - 6)}

Decision Rules & Collaboration Principles:
1. Team Priority First: Decisions must aim for global task victory, prioritizing coordination with teammates' tactical tendencies (e.g., if teammates focus fire on a specific enemy unit, coordinate the attack; if teammates retreat, provide cover).
2. Observation-Driven Decision Making: Judge strictly based on input message data; Do not assume unobserved information (e.g., enemy units outside vision, status of agents that have not sent messages).
3. Action Effectiveness Constraints: Only select actions from your own list of executable actions; do not perform operations beyond your capabilities.
4. Historical Action Reference: Adjust strategies based on the effect of teammates' previous K steps of actions (e.g., if teammates continuously attack an enemy unit but fail to destroy it, evaluate whether support is needed; if teammates' movement paths show advancement towards the target area, follow up).

Output Format Requirements:
Output only contains SMAC directly executable action commands in the above action list, and the reason why you made the choice.
{
                "action": (action index),
                "reason": (why you do this action)
}

Input & Output Examples: ......

Note:
1. Your answer must follow the given format without any extra output.
2. You must select actions only from the provided list of available actions; actions outside this list are prohibited.
\end{lstlisting}
    \caption{Example of LLM Prompt for Multi-Agent Communication}
    \label{fig:llm-prompt}
\end{figure}

\subsection{Experiments}
We evaluate the performance of the algorithm-LLM hybrid decision-making mechanism on the SMAC benchmark, which comprises N-1 algorithm-controlled agents (A-Agents) and one LLM-controlled agent (L-Agent).
Given the overhead of real-time interaction, we do not incorporate LLMs into the training of our method, Herm, and only replace one team member with an L-Agent during testing.

We present the LLM’s reasoning process during the decision-making phase, as illustrated in Fig.\ref{fig:llm-message-output}. It can be observed that the LLM is able to perform reasoning based on the decoded and processed natural language messages and make corresponding decisions.

\begin{figure}[ht]
    \centering

    \begin{lstlisting}[style=promptstyle]
No enemy is visible to me or any ally based on the messages. 
I can see Ally_0 to my northwest (relative position x -1.00~-0.80, y -0.60~-0.40). 
Agent 0's message indicates they see me to their east (relative x -0.20~0.00, y -0.60~-0.40), so I am west of them. 
To regroup and explore actively (as per principle 5), moving East (action 4) will bring me closer to Ally_0 and potentially towards the area where other allies might be.
 Because Agent 0 sees Ally_1 to their northwest, this choice aligns with team coordination and active participation.
    \end{lstlisting}
        \caption{Example of LLM Reasoning}
    \label{fig:llm-message-output}
    \end{figure}

Tab.\ref{tab:llm} presents the hybrid decision-making performance of the algorithm with different LLMs.
On the one hand, Herm maintains a certain level of performance when integrated with the three LLMs, which validates the generalizability of the proposed hyperdimensional symbolic messages in hybrid decision-making.

On the other hand, notable performance discrepancies are observed in collaboration with different LLMs, a phenomenon that can be attributed to the intrinsic decision-making competence of LLMs.
In terms of the dimensionality of the hyperdimensional space $\mathcal{H}$, performance exhibits an upward trend with increasing dimensionality. 
This is because a larger hyperdimensional space facilitates accurate semantic transformation. 
In addition, the upward trend gradually slows down under high-dimensional settings: minor transformation errors are corrected by other message sources during LLM decision-making, leading to a reduced impact on decision outcomes.
\begin{table}[h]
    \centering
    \caption{Performance of Algorithm-LLM Hybrid Decision Making}
    \begin{tabular}{cccccc}
    \toprule
       $\mathcal{H}$ & Task & Herm & DeepSeek-V3.2 & GLM-4-7 & Doubao-Seed-1.6\\
    \midrule
      \multirow{3}{*}{256}& 5z\_vs\_1ul & 79.0$\pm$0.2 & 68.4$\pm$3.7 & 31.6$\pm$3.6 & 53.8$\pm$3.5\\
      & 2s3z\_vs\_1ul  & 81.9$\pm$0.6 & 64.3$\pm$3.3 & 35.6$\pm$4.8 & 60.5$\pm$3.1\\
      & 1c4z\_vs\_1ul  & 75.1$\pm$1.5 & 52.6$\pm$3.4 & 28.5$\pm$3.3 & 44.6$\pm$4.2\\
        \midrule
      \multirow{3}{*}{512}& 5z\_vs\_1ul & 80.3$\pm$0.8 & 73.4$\pm$2.2 & 34.3$\pm$1.5 & 58.0$\pm$2.6\\
      & 2s3z\_vs\_1ul  & 80.6$\pm$0.5 & 66.8$\pm$3.7 & 37.2$\pm$3.4 & 59.8$\pm$2.8\\
      & 1c4z\_vs\_1ul  & 76.5$\pm$0.8 & 55.2$\pm$2.8 & 30.5$\pm$3.1 & 46.3$\pm$3.5\\
    \midrule
      \multirow{3}{*}{1024}& 5z\_vs\_1ul & 79.2$\pm$0.5 & 74.1$\pm$1.9 & 36.4$\pm$1.0 & 59.2$\pm$2.4\\
      & 2s3z\_vs\_1ul  & 81.0$\pm$1.0 & 68.0$\pm$1.2 & 38.8$\pm$2.6 & 61.3$\pm$2.3\\
      & 1c4z\_vs\_1ul  & 75.8$\pm$0.7 & 56.1$\pm$2.4 & 30.9$\pm$2.8 & 47.1$\pm$1.5\\
    \bottomrule
    \end{tabular}
    \label{tab:llm}
\end{table}

\section{Algorithm Pseudocode}
\label{Pseudocode}
We introduce the decision-making process of a single agent (Alg. \ref{alg:Herm}) and the algorithm training process (Alg. \ref{alg:qmix}) in Herm, respectively. The overall training architecture employs QMIX.

\begin{algorithm}[!htb]
\caption{Training Process in Herm}
\label{alg:qmix}

\KwIn{
    $N$ (number of agents), $O$ (observation space), $A_i$ (action space of agent $i$), 
    $\gamma$ (discount factor), $\alpha$ (learning rate), $\tau$ (soft update coefficient)
}
\KwOut{Policies $\pi_1, \dots, \pi_N$ for each agent}

\BlankLine
\textbf{Initialization:} \\
\Indp
    Agent Q-networks $Q_{local}(o_i, a_i; \theta)$ with parameters $\theta$ \\
    Mixing network ${Mixing}(q_1, \dots, q_N; \beta)$ with $\beta$ (non-negative weights $\beta$) \\
    Target Q-networks $Q_{local}^{\text{target}}$ and target mixing network ${Mixing}^{\text{target}}$: 
    \Indp
        $\theta^- \leftarrow \theta$, $\beta^- \leftarrow \beta$\\
    \Indm
    Replay buffer $D$ (stores tuples $(s,\boldsymbol{o}, \boldsymbol{a}, r, s', \text{done})$\\
    Exploration strategy $\epsilon$ (e.g., $\epsilon$-greedy)
\Indm

\BlankLine
\For{each training episode}{
    Initialize state $s$
    \While{not terminated}{
        Each agent $i$ selects $a_i = \pi_i(o_i)$ via $Q_{\mathrm{local}}(o_i, \cdot; \theta)$ and $\epsilon$ 
        \tcp*{see Alg. 2}
        Execute $\mathbf{a}$, get reward $r$ and next state $s'$\\
        Store $(s, \boldsymbol{o},\boldsymbol{a}, r, s', \mathrm{done})$ in $D$
        $s \leftarrow s'$
    }
}
\For{each parameter update}{
    Sample batch $\{(s_j, \boldsymbol{o}_j, \boldsymbol{a}_j, r_j, s'_j, \mathrm{done}_j)\}_{j=1}^B$ from $D$\\
    \For{each sample $j$}{
        Each agent $i$ selects action: $a'_i = \arg\max_{a'} Q_{\mathrm{local}}(o'_{i,j}, a'; \theta)$\\
        Agent target: $y_i^j = Q_{\mathrm{local}}^{\mathrm{target}}(o'_{i,j}, a'_i; \theta^-)$\\
        Joint target: 
        \[
            Y_j = r_j + \gamma (1 - \mathrm{done}_j) \cdot {Mixing}^{\mathrm{target}}(y_1^j, \dots, y_N^j; \beta^-)
        \]
        Current agent Q-values: $q_i^j = Q_{\mathrm{local}}(o_{i,j}, a_{i,j}; \theta)$\\
        Current joint value: $Q_{\mathrm{total},j} = {Mixing}(q_1^j, \dots, q_N^j; \beta)$
    }
    Loss: $L = \frac{1}{B} \sum_{j=1}^B (Y_j - Q_{\mathrm{total},j})^2$\\
    Update $\theta$ and $\beta$ via gradient descent (ensure $\beta \geq 0$)\\
    Soft update targets:
    $\theta^- \leftarrow \tau \theta + (1 - \tau) \theta^-$
    $\beta^- \leftarrow \tau \beta + (1 - \tau) \beta^-$
}

\BlankLine
\Return Policies $\pi_i(a_i | s) = \arg\max_{a_i} Q_{local}(o_i, a_i; \theta)$
\end{algorithm}

\begin{algorithm}[!htb] 
\caption{Decision making of Agent $i$ in Herm} 
\label{alg:Herm}

\textbf{Initialize: }
Local Q Network $Q_{local}$, Observation Projection $\Phi$, Action Projection $\Psi$, Trajectory Encoder $GRU$, Attribute Dictionary $\mathcal{A}$, Value Dictionary $\mathcal{V}$, Action Dictionary $\mathcal{A}ct$, Time Dictionary $\mathcal{T}$,
Length of action sequence $\mathcal{W}$

\tcp{1. Message Encoding}
\tcp{1.1 Observation Transformation}
$HV_O \leftarrow Empty$ 

\ForEach{$\mathbf{e}_j \in Entity$}{ 
    $o_i^j \leftarrow f(o_i )$ \tcp*{Entity-wised Extraction}
    $HV_o^j \leftarrow Empty$ \\
    \ForEach{$(attr, val) \in o_i^j$}
    {
        $HV_{attr} \leftarrow find(\mathcal{A},attr)$ \\
        $HV_{val} \leftarrow find(\mathcal{V},attr)$ \\
        $HV_o^j = HV_o^j \oplus(HV_{attr} \otimes HV_{val})$
    }
    $HV_o^j = HV_o^j \otimes HV_{id}^j$ \\
    $HV_O = HV_O \oplus HV_o^j$
}

\tcp{1.2 Action Sequence Transformation}
$HV_A \leftarrow Empty$ 

\ForEach{$ w \in \mathcal{W}$}{ 
        $HV_{a}^{-w} \leftarrow find(\mathcal{A}ct,a_i^{-w})$ \\
        $HV_{t}^{-w} \leftarrow find(\mathcal{T},-w)$ \\
        $HV_A = HV_A \oplus (HV_{a}^{-w} \otimes HV_{t}^{-w})$
}

\tcp{2. Message Fusion}
${\alpha}_i = cosine\_sim(HV_O, HV_O^{-i})$\\
$\boldsymbol{\alpha} = softmax({\alpha}_i)$\\
$\boldsymbol{\dot\alpha} = softmax({-\alpha}_i)$\\
$c_o^i = \Phi(\sum_{j}\boldsymbol{\dot\alpha}_{i,j} * (HV_O^{i} \otimes HV_O^{i,j}))$\\
$c_a^i = \Psi(\sum_{j}\boldsymbol{\alpha}_{i,j} * (HV_A^{i} \otimes HV_A^{i,j}))$\\
$\boldsymbol{c}_i=(c_o^i, c_a^i)$

\tcp{3. Trajectory Encoding}
$h_i^t \leftarrow GRU(h_i^{t-1},o_i^t, a_i^{t-1})$ \\
\Return $Q_{local}(h_i^t, \boldsymbol{c}_i)$
\end{algorithm}


\clearpage
\input{checklist.tex}

\end{document}

%% file: checklist.tex
\section*{NeurIPS Paper Checklist}

\begin{enumerate}

\item {\bf Claims}
    \item[] Question: Do the main claims made in the abstract and introduction accurately reflect the paper's contributions and scope?
    \item[] Answer: \answerYes{} 
    \item[] Justification: see Abstract and Introduction.

\item {\bf Limitations}
    \item[] Question: Does the paper discuss the limitations of the work performed by the authors?
    \item[] Answer: \answerYes{} 
    \item[] Justification: see Conclusion.

\item {\bf Theory assumptions and proofs}
    \item[] Question: For each theoretical result, does the paper provide the full set of assumptions and a complete (and correct) proof?
    \item[] Answer: \answerYes{} 
    \item[] Justification: see Appendix.

    \item {\bf Experimental result reproducibility}
    \item[] Question: Does the paper fully disclose all the information needed to reproduce the main experimental results of the paper to the extent that it affects the main claims and/or conclusions of the paper (regardless of whether the code and data are provided or not)?
    \item[] Answer: \answerYes{} 
    \item[] Justification: see Appendix.

\item {\bf Open access to data and code}
    \item[] Question: Does the paper provide open access to the data and code, with sufficient instructions to faithfully reproduce the main experimental results, as described in supplemental material?
    \item[] Answer: \answerYes{} 
    \item[] Justification: See Supplementary Material.

\item {\bf Experimental setting/details}
    \item[] Question: Does the paper specify all the training and test details (e.g., data splits, hyperparameters, how they were chosen, type of optimizer) necessary to understand the results?
    \item[] Answer: \answerYes{} 
    \item[] Justification: see Appendix.

\item {\bf Experiment statistical significance}
    \item[] Question: Does the paper report error bars suitably and correctly defined or other appropriate information about the statistical significance of the experiments?
    \item[] Answer: \answerYes{} 
    \item[] Justification: see Section Experiments.

\item {\bf Experiments compute resources}
    \item[] Question: For each experiment, does the paper provide sufficient information on the computer resources (type of compute workers, memory, time of execution) needed to reproduce the experiments?
    \item[] Answer: \answerYes{} 
    \item[] Justification: see Appendix.
    
\item {\bf Code of ethics}
    \item[] Question: Does the research conducted in the paper conform, in every respect, with the NeurIPS Code of Ethics \url{https://neurips.cc/public/EthicsGuidelines}?
    \item[] Answer: \answerYes{} 
    \item[] Justification: See Supplementary Material.

\item {\bf Broader impacts}
    \item[] Question: Does the paper discuss both potential positive societal impacts and negative societal impacts of the work performed?
    \item[] Answer: \answerYes{} 
    \item[] Justification: See Conclusion.
    
\item {\bf Safeguards}
    \item[] Question: Does the paper describe safeguards that have been put in place for responsible release of data or models that have a high risk for misuse (e.g., pre-trained language models, image generators, or scraped datasets)?
    \item[] Answer: \answerNA{} 
    \item[] Justification: Do not involve.

\item {\bf Licenses for existing assets}
    \item[] Question: Are the creators or original owners of assets (e.g., code, data, models), used in the paper, properly credited and are the license and terms of use explicitly mentioned and properly respected?
    \item[] Answer: \answerYes{}{} 
    \item[] Justification: see Appendix.

\item {\bf New assets}
    \item[] Question: Are new assets introduced in the paper well documented and is the documentation provided alongside the assets?
    \item[] Answer: \answerNA{} 
    \item[] Justification: Do not involve.

\item {\bf Crowdsourcing and research with human subjects}
    \item[] Question: For crowdsourcing experiments and research with human subjects, does the paper include the full text of instructions given to participants and screenshots, if applicable, as well as details about compensation (if any)? 
    \item[] Answer: \answerNA{} 
    \item[] Justification: Do not involve.

\item {\bf Institutional review board (IRB) approvals or equivalent for research with human subjects}
    \item[] Question: Does the paper describe potential risks incurred by study participants, whether such risks were disclosed to the subjects, and whether Institutional Review Board (IRB) approvals (or an equivalent approval/review based on the requirements of your country or institution) were obtained?
    \item[] Answer: \answerNA{} 
    \item[] Justification: Do not involve.

\item {\bf Declaration of LLM usage}
    \item[] Question: Does the paper describe the usage of LLMs if it is an important, original, or non-standard component of the core methods in this research? Note that if the LLM is used only for writing, editing, or formatting purposes and does \emph{not} impact the core methodology, scientific rigor, or originality of the research, declaration is not required.
    \item[] Answer: \answerYes{} 
    \item[] Justification: see Appendix H.

\end{enumerate}

%% file: neurips_2026.bib
@ARTICLE{HRR,
  author={Plate, T.A.},
  journal={IEEE Transactions on Neural Networks}, 
  title={Holographic reduced representations}, 
  year={1995},
  volume={6},
  number={3},
  pages={623-641},
}

@InProceedings{BSC,
author="Kanerva, Pentti",
title="Binary spatter-coding of ordered K-tuples",
booktitle="Artificial Neural Networks --- ICANN 96",
year="1996",
publisher="Springer Berlin Heidelberg",
pages="869--873",
}

@article{vbc,
  title={Efficient communication in multi-agent reinforcement learning via variance based control},
  author={Zhang, Sai Qian and Zhang, Qi and Lin, Jieyu},
  journal={Advances in neural information processing systems},
  volume={32},
  year={2019}
}

@article{TMC,
  title={Succinct and robust multi-agent communication with temporal message control},
  author={Zhang, Sai Qian and Zhang, Qi and Lin, Jieyu},
  journal={Advances in neural information processing systems},
  volume={33},
  pages={17271--17282},
  year={2020}
}

@inproceedings{tarmac,
  title={Tarmac: Targeted multi-agent communication},
  author={Das, Abhishek and Gervet, Th{\'e}ophile and Romoff, Joshua and Batra, Dhruv and Parikh, Devi and Rabbat, Mike and Pineau, Joelle},
  booktitle={International Conference on machine learning},
  pages={1538--1546},
  year={2019},
  organization={PMLR}
}

@inproceedings{CACOM,
  title={Context-aware Communication for Multi-agent Reinforcement Learning},
  author={Li, Xinran and Zhang, Jun},
  booktitle={Proceedings of the 23rd International Conference on Autonomous Agents and Multiagent Systems},
  pages={1156--1164},
  year={2024}
}

@article{GP2020,
  author       = {Rupert Mitchell and
                  Jan Blumenkamp and
                  Amanda Prorok},
  title        = {Gaussian Process Based Message Filtering for Robust Multi-Agent Cooperation
                  in the Presence of Adversarial Communication},
  journal      = {CoRR},
  volume       = {abs/2012.00508},
  year         = {2020},
  eprinttype    = {arXiv},
  eprint       = {2012.00508},
}

@inproceedings{MIS2022,
  author       = {Wanqi Xue and
                  Wei Qiu and
                  Bo An and
                  Zinovi Rabinovich and
                  Svetlana Obraztsova and
                  Chai Kiat Yeo},
  title        = {Mis-spoke or mis-lead: Achieving Robustness in Multi-Agent Communicative
                  Reinforcement Learning},
  booktitle    = {International Conference on Autonomous Agents and Multiagent
                  Systems},
  pages        = {1418--1426},
  year         = {2022},
}

@inproceedings{AME2023,
  title={Certifiably robust policy learning against adversarial multi-agent communication},
  author={Sun, Yanchao and Zheng, Ruijie and Hassanzadeh, Parisa and Liang, Yongyuan and Feizi, Soheil and Ganesh, Sumitra and Huang, Furong},
  booktitle={The Eleventh International Conference on Learning Representations},
  year={2022}
}

@inproceedings{ADMAC2024,
  title={Robust communicative multi-agent reinforcement learning with active defense},
  author={Yu, Lebin and Qiu, Yunbo and Yao, Quanming and Shen, Yuan and Zhang, Xudong and Wang, Jian},
  booktitle={Proceedings of the AAAI Conference on Artificial Intelligence},
  volume={38},
  number={16},
  pages={17575--17582},
  year={2024}
}

@article{DMAC2025,
  author       = {Xuyan Ma and
                  Yawen Wang and
                  Junjie Wang and
                  Xiaofei Xie and
                  Boyu Wu and
                  Shoubin Li and
                  Fanjiang Xu and
                  Qing Wang},
  title        = {Robust Multi-agent Communication Based on Decentralization-Oriented
                  Adversarial Training},
  journal      = {CoRR},
  volume       = {abs/2504.21278},
  year         = {2025},
}

@inproceedings{MDGI2024,
  title={Learning efficient and robust multi-agent communication via graph information bottleneck},
  author={Ding, Shifei and Du, Wei and Ding, Ling and Guo, Lili and Zhang, Jian},
  booktitle={Proceedings of the AAAI Conference on Artificial Intelligence},
  volume={38},
  number={16},
  pages={17346--17353},
  year={2024}
}

@article{CroMAC2024,
  title={Robust cooperative multi-agent reinforcement learning via multi-view message certification},
  author={Yuan, Lei and Jiang, Tao and Li, Lihe and Chen, Feng and Zhang, Zongzhang and Yu, Yang},
  journal={Science China Information Sciences},
  volume={67},
  number={4},
  pages={142102},
  year={2024},
  publisher={Springer}
}

@article{MA3C2024,
  title={Communication-robust multi-agent learning by adaptable auxiliary multi-agent adversary generation},
  author={Yuan, Lei and Chen, Feng and Zhang, Zongzhang and Yu, Yang},
  journal={Frontiers of Computer Science},
  volume={18},
  number={6},
  pages={186331},
  year={2024},
  publisher={Springer}
}

@inproceedings{HDPG2022,
  title={Hdpg: Hyperdimensional policy-based reinforcement learning for continuous control},
  author={Ni, Yang and Issa, Mariam and Abraham, Danny and Imani, Mahdi and Yin, Xunzhao and Imani, Mohsen},
  booktitle={Proceedings of the 59th ACM/IEEE Design Automation Conference},
  pages={1141--1146},
  year={2022}
}

@inproceedings{QHD2023,
  title={Efficient off-policy reinforcement learning via brain-inspired computing},
  author={Ni, Yang and Abraham, Danny and Issa, Mariam and Kim, Yeseong and Mercati, Pietro and Imani, Mohsen},
  booktitle={Proceedings of the Great Lakes Symposium on VLSI 2023},
  pages={449--453},
  year={2023}
}

@inproceedings{ROBOT2024,
  title={Brain-Inspired Hyperdimensional Computing in the Wild: Lightweight Symbolic Learning for Sensorimotor Controls of Wheeled Robots},
  author={Kwon, Hyukjun and Kim, Kangwon and Lee, Junyoung and Lee, Hyunsei and Kim, Jiseung and Kim, Jinhyung and Kim, Taehyung and Kim, Yongnyeon and Ni, Yang and Imani, Mohsen and others},
  booktitle={2024 IEEE International Conference on Robotics and Automation (ICRA)},
  pages={5176--5182},
  year={2024},
  organization={IEEE}
}

@ARTICLE{Cyber2025,
  author={Issa, Mariam Ali and Chen, Hanning and Wang, Junyao and Imani, Mohsen},
  journal={IEEE Transactions on Computer-Aided Design of Integrated Circuits and Systems},
  title={CyberRL: Brain-Inspired Reinforcement Learning for Efficient Network Intrusion Detection},
  year={2025},
  volume={44},
  number={1},
  pages={241-250},
}

@article{HDCB2025,
  author       = {Marco Angioli and
                  Antonello Rosato and
                  Marcello Barbirotta and
                  Rocco Martino and
                  Francesco Menichelli and
                  Mauro Olivieri},
  title        = {{HD-CB:} The First Exploration of Hyperdimensional Computing for Contextual
                  Bandits Problems},
  journal      = {CoRR},
  volume       = {abs/2501.16863},
  year         = {2025},
  eprinttype    = {arXiv},
  eprint       = {2501.16863},
}

@inproceedings{SMAC,
  author       = {Mikayel Samvelyan and
                  Tabish Rashid and
                  Christian Schr{\"{o}}der de Witt and
                  Gregory Farquhar and
                  Nantas Nardelli and
                  Tim G. J. Rudner and
                  Chia{-}Man Hung and
                  Philip H. S. Torr and
                  Jakob N. Foerster and
                  Shimon Whiteson},
  title        = {The StarCraft Multi-Agent Challenge},
  booktitle    = {Proceedings of the 18th International Conference on Autonomous Agents
                  and MultiAgent Systems},
  pages        = {2186--2188},
  year         = {2019},
}

@inproceedings{SMACV2,
  author       = {Benjamin Ellis and
                  Jonathan Cook and
                  Skander Moalla and
                  Mikayel Samvelyan and
                  Mingfei Sun and
                  Anuj Mahajan and
                  Jakob N. Foerster and
                  Shimon Whiteson},
  title        = {SMACv2: An Improved Benchmark for Cooperative Multi-Agent Reinforcement
                  Learning},
  booktitle    = {Advances in Neural Information Processing Systems 36: Annual Conference
                  on Neural Information Processing Systems},
  year         = {2023},
}

@inproceedings{NDQ,
  author       = {Tonghan Wang and
                  Jianhao Wang and
                  Chongyi Zheng and
                  Chongjie Zhang},
  title        = {Learning Nearly Decomposable Value Functions Via Communication Minimization},
  booktitle    = {8th International Conference on Learning Representations, {ICLR} 2020,
                  Addis Ababa, Ethiopia, April 26-30, 2020},
  year         = {2020},
}

@inproceedings{LBF,
   title={Benchmarking Multi-Agent Deep Reinforcement Learning Algorithms in Cooperative Tasks},
   author={Georgios Papoudakis and Filippos Christianos and Lukas Schäfer and Stefano V. Albrecht},
   booktitle = {Proceedings of the Neural Information Processing Systems Track on Datasets and Benchmarks (NeurIPS)},
   year={2021},
   openreview = {https://openreview.net/forum?id=cIrPX-Sn5n},
}

@inproceedings{T2MAC,
  title={T2mac: Targeted and trusted multi-agent communication through selective engagement and evidence-driven integration},
  author={Sun, Chuxiong and Zang, Zehua and Li, Jiabao and Li, Jiangmeng and Xu, Xiao and Wang, Rui and Zheng, Changwen},
  booktitle={Proceedings of the AAAI Conference on Artificial Intelligence},
  volume={38},
  number={13},
  pages={15154--15163},
  year={2024}
}

@inproceedings{MAIC,
  title={Multi-agent incentive communication via decentralized teammate modeling},
  author={Yuan, Lei and Wang, Jianhao and Zhang, Fuxiang and Wang, Chenghe and Zhang, Zongzhang and Yu, Yang and Zhang, Chongjie},
  booktitle={Proceedings of the AAAI conference on artificial intelligence},
  volume={36},
  number={9},
  pages={9466--9474},
  year={2022}
}

@inproceedings{TGCNET,
  title={Bridging training and execution via dynamic directed graph-based communication in cooperative multi-agent systems},
  author={Zhang, Zhuohui and He, Bin and Cheng, Bin and Li, Gang},
  booktitle={Proceedings of the AAAI Conference on Artificial Intelligence},
  volume={39},
  number={22},
  pages={23395--23403},
  year={2025}
}

@inproceedings{CODE,
  author       = {Shoucheng Song and
                  Youfang Lin and
                  Sheng Han and
                  Chang Yao and
                  Hao Wu and
                  Shuo Wang and
                  Kai Lv},
  title        = {CoDe: Communication Delay-Tolerant Multi-Agent Collaboration via Dual
                  Alignment of Intent and Timeliness},
  booktitle    = {AAAI-25, Sponsored by the Association for the Advancement of Artificial
                  Intelligence},
  pages        = {23304--23312},
  year         = {2025},
}

@ARTICLE{MAGI,
  author={Ding, Shifei and Du, Wei and Ding, Ling and Zhang, Jian and Guo, Lili and An, Bo},
  journal={IEEE Transactions on Pattern Analysis and Machine Intelligence}, 
  title={Robust Multi-Agent Communication With Graph Information Bottleneck Optimization}, 
  year={2024},
  volume={46},
  number={5},
  pages={3096-3107}}

@inproceedings{imani2018hierarchical,
  title={Hierarchical hyperdimensional computing for energy efficient classification},
  author={Imani, Mohsen and Huang, Chenyu and Kong, Deqian and Rosing, Tajana},
  booktitle={Proceedings of the 55th Annual Design Automation Conference},
  pages={1--6},
  year={2018}
}

@inproceedings{zhao2022fedhd,
  title={Fedhd: federated learning with hyperdimensional computing},
  author={Zhao, Quanling and Lee, Kai and Liu, Jeffrey and Huzaifa, Muhammad and Yu, Xiaofan and Rosing, Tajana},
  booktitle={Proceedings of the 28th annual international conference on mobile computing and networking},
  pages={791--793},
  year={2022}
}

@article{FGSM,
  title={Explaining and harnessing adversarial examples},
  author={Goodfellow, Ian J and Shlens, Jonathon and Szegedy, Christian},
  journal={arXiv preprint arXiv:1412.6572},
  year={2014}
}

@article{PGD,
  title={Towards deep learning models resistant to adversarial attacks},
  author={Madry, Aleksander and Makelov, Aleksandar and Schmidt, Ludwig and Tsipras, Dimitris and Vladu, Adrian},
  journal={arXiv preprint arXiv:1706.06083},
  year={2017}
}

@article{li2024language,
  title={Language grounded multi-agent reinforcement learning with human-interpretable communication},
  author={Li, Huao and Nourkhiz Mahjoub, Hossein and Chalaki, Behdad and Tadiparthi, Vaishnav and Lee, Kwonjoon and Moradi Pari, Ehsan and Lewis, Charles and Sycara, Katia},
  journal={Advances in Neural Information Processing Systems},
  volume={37},
  pages={87908--87933},
  year={2024}
}

@article{VTB,
    author = {Gosmann, Jan and Eliasmith, Chris},
    title = {Vector-Derived Transformation Binding: An Improved Binding Operation for Deep Symbol-Like Processing in Neural Networks},
    journal = {Neural Computation},
    volume = {31},
    number = {5},
    pages = {849-869},
    year = {2019},
    month = {05},
    issn = {0899-7667},
}

@book{FHRR,
  title={Holographic Reduced Representation: Distributed representation for cognitive structures},
  author={Plate, Tony A},
  volume={150},
  year={2003},
  publisher={CSLI Publications Stanford}
}

@article{MADDPG,
  title={Multi-agent actor-critic for mixed cooperative-competitive environments},
  author={Lowe, Ryan and Wu, Yi I and Tamar, Aviv and Harb, Jean and Pieter Abbeel, OpenAI and Mordatch, Igor},
  journal={Advances in neural information processing systems},
  volume={30},
  year={2017}
}

@article{mappo,
  title={The surprising effectiveness of ppo in cooperative multi-agent games},
  author={Yu, Chao and Velu, Akash and Vinitsky, Eugene and Gao, Jiaxuan and Wang, Yu and Bayen, Alexandre and Wu, Yi},
  journal={Advances in neural information processing systems},
  volume={35},
  pages={24611--24624},
  year={2022}
}

@inproceedings{pymarlzooplus,
author = {Papadopoulos, George and Kontogiannis, Andreas and Papadopoulou, Foteini and Poulianou, Chaido and Koumentis, Ioannis and Vouros, George},
title = {An Extended Benchmarking of Multi-Agent Reinforcement Learning Algorithms in Complex Fully Cooperative Tasks},
year = {2025},
isbn = {9798400714269},
booktitle = {Proceedings of the 24th International Conference on Autonomous Agents and Multiagent Systems},
pages = {1613–1622},
numpages = {10},
}

@inproceedings{RPG,
  title={From general relation patterns to task-specific decision-making in continual multi-agent coordination},
  author={Yao, Chang and Lin, Youfang and Song, Shoucheng and Wu, Hao and Ma, Yuqing and Han, Sheng and Lv, Kai},
  booktitle={Proceedings of the Thirty-Fourth International Joint Conference on Artificial Intelligence},
  pages={6821--6829},
  year={2025}
}

@article{MATTAR,
  title={Multi-agent policy transfer via task relationship modeling},
  author={Qin, Rongjun and Chen, Feng and Wang, Tonghan and Yuan, Lei and Wu, Xiaoran and Kang, Yipeng and Zhang, Zongzhang and Zhang, Chongjie and Yu, Yang},
  journal={Science China Information Sciences},
  volume={67},
  number={8},
  pages={182101},
  year={2024},
  publisher={Springer}
}

@article{OPT,
  title={Interaction pattern disentangling for multi-agent reinforcement learning},
  author={Liu, Shunyu and Song, Jie and Zhou, Yihe and Yu, Na and Chen, Kaixuan and Feng, Zunlei and Song, Mingli},
  journal={IEEE Transactions on Pattern Analysis and Machine Intelligence},
  volume={46},
  number={12},
  pages={8157--8172},
  year={2024},
  publisher={IEEE}}

@inproceedings{ODIS,
  title={Discovering generalizable multi-agent coordination skills from multi-task offline data},
  author={Zhang, Fuxiang and Jia, Chengxing and Li, Yi-Chen and Yuan, Lei and Yu, Yang and Zhang, Zongzhang},
  booktitle={The Eleventh International Conference on Learning Representations},
  year={2023}
}

@article{spearman,
 ISSN = {00029556},
 URL = {http://www.jstor.org/stable/1412159},
 author = {C. Spearman},
 journal = {The American Journal of Psychology},
 number = {1},
 pages = {72--101},
 publisher = {University of Illinois Press},
 title = {The Proof and Measurement of Association between Two Things},
 urldate = {2026-10-16},
 volume = {15},
 year = {1904}
}

@misc{dear,
      title={Group Perspective Matters: Regulating Debate Relationships Can Mitigate Blind Conformity in Multi-Agent Debate}, 
      author={Hao Wu and Shoucheng Song and Chang Yao and Haoyu Wang and Huaiyu Wan and Youfang Lin and Kai Lv},
      year={2026},
      eprint={2608.03648},
      archivePrefix={arXiv},
      primaryClass={cs.MA},
      url={https://arxiv.org/abs/2608.03648}, 
}

@inproceedings{AIM,
  title={Think how your teammates think: Active inference can benefit decentralized execution},
  author={Wu, Hao and Song, Shoucheng and Yao, Chang and Han, Sheng and Wan, Huaiyu and Lin, Youfang and Lv, Kai},
  booktitle={Proceedings of the AAAI Conference on Artificial Intelligence},
  volume={40},
  number={35},
  pages={29749--29757},
  year={2026}
}

@inproceedings{yao1,
  title     = {From General Relation Patterns to Task-Specific Decision-Making in Continual Multi-Agent Coordination},
  author    = {Yao, Chang and Lin, Youfang and Song, Shoucheng and Wu, Hao and Ma, Yuqing and Han, Sheng and Lv, Kai},
  booktitle = {Proceedings of the Thirty-Fourth International Joint Conference on
               Artificial Intelligence, {IJCAI-25}},
  publisher = {International Joint Conferences on Artificial Intelligence Organization},
  editor    = {James Kwok},
  pages     = {6821--6829},
  year      = {2025},
  month     = {8},
  note      = {Main Track},
  doi       = {10.24963/ijcai.2025/759},
  url       = {https://doi.org/10.24963/ijcai.2025/759},
}

@inproceedings{yao26direct,
  title={DIRECT: Decentralized Intention and Latent Rule Emergence for Multi-Agent Cooperation Under Intermittent Communication},
  author={Yao, Enguang and Liu, Qidong and Sun, Jialu and Huang, Kaibo and Xu, Mingliang},
  booktitle={Proceedings of the Thirty-Fifth International Joint Conference on Artificial Intelligence, IJCAI},
  volume={26},
  pages={437--445}
}

@inproceedings{jiang2026budget,
  title={Budget-Efficient Attacks and Robustness Training for Cooperative MARL},
  author={Jiang, Junyong and Yuan, Xin and Lin, Longhe and Li, Songze and Dong, Lu},
  booktitle={Forty-third International Conference on Machine Learning},
  year={2026}
}

@inproceedings{chu2026value,
  title={Value Aggregation with Uncertainty in Online Decentralized MARL},
  author={Chu, Ziyue and Stella, Leonardo},
  booktitle={Forty-third International Conference on Machine Learning},
  year={2026}
}

@inproceedings{yao2,
  title={Role-Level Inductive Bias for Cross-Task Generalization in Multi-Agent Reinforcement Learning},
  author={Yao, Chang and Lin, Youfang and Song, Shoucheng and Wu, Hao and Yang, Shengkun and Ma, Yuqing and Lv, Kai},
  booktitle={Forty-third International Conference on Machine Learning}
}
